\newif\ifconf

\documentclass[11pt,letterpaper]{article}
\ifconf
\usepackage[notes=false,later=false]{dtrt}
\else
\usepackage[notes=false,later=false]{dtrt}
\fi

\usepackage{fullpage}

\usepackage{amssymb,amsfonts,amsmath}

\usepackage{amsthm}

\usepackage{times}
\usepackage{microtype}
\usepackage{tikz}
  \usetikzlibrary{matrix,tikzmark}
\usepackage{graphicx}
\usepackage{verbatim}
\usepackage{array}
\usepackage{arydshln}
\usepackage{multirow}
\usepackage{paralist}
\usepackage{enumitem}
  \setlist[itemize]{leftmargin=*}
  \setlist[enumerate]{leftmargin=*}
\usepackage[capitalise]{cleveref}
  \crefname{step}{Step}{Steps}
  \crefname{item}{Item}{Items}
  \crefname{condition}{Condition}{Conditions}
  \crefname{enumi}{Step}{Steps}
  \crefname{case}{Case}{Cases}
  \creflabelformat{equation}{#2\textup{#1}#3}
\usepackage{dsfont}
\usepackage{url}
\usepackage{braket}
\usepackage{mathrsfs}
\usepackage{soul}
\usepackage{mdframed}
\usepackage{mathtools}
\usepackage{bbm}
\usepackage{setspace}
\usepackage{tabu}
\usepackage{multirow}
\usepackage{adjustbox}
\usepackage{booktabs}
\usepackage[binary-units=true]{siunitx}
\usepackage[margin=4mm,small,labelfont=bf]{caption}
\usepackage{centernot}
\usepackage{environ}
\usepackage{bbm}

\usepackage{algorithm}
\usepackage[noend]{algorithmic}

\floatname{algorithm}{Reduction}
\crefname{algorithm}{reduction}{reductions}
\Crefname{algorithm}{Reduction}{Reductions}

\usepackage[giveninits,backend=bibtex,style=alphabetic,maxnames=6,maxalphanames=6]{biblatex}

\bibliography{title-long,references}

\theoremstyle{plain} 
\newtheorem{itheorem}{Theorem}
\newtheorem{idefinition}{Definition}
\newtheorem{theorem}{Theorem}[section]
\newtheorem{lemma}[theorem]{Lemma}
\newtheorem{proposition}[theorem]{Proposition}
\newtheorem{claim}[theorem]{Claim}
\newtheorem{definition}[theorem]{Definition}
\newtheorem{corollary}[theorem]{Corollary}

\newtheorem{openproblem}{Open Problem}

\theoremstyle{definition} 
\newtheorem{construction}[theorem]{Construction}

\theoremstyle{remark}

\newcommand{\N}{\mathbb{N}}
\newcommand{\Field}{\mathbb{F}}

\newcommand{\Bits}{\{0,1\}}

\newcommand{\Code}{C}

\newcommand{\Domain}{\text{Dom}}

\newcommand{\BlockLength}{\ell}

\newcommand{\ReedMuller}{\mathrm{RM}}

\newcommand{\Distance}{d}

\newcommand{\HammingDist}{\Delta}

\newcommand{\rev}{\mathsf{rev}}

\newcommand{\revvec}[1]{\vec{#1}_{\rev}}

\newcommand{\Alphabet}{\Sigma}
\newcommand{\Prover}{\mathcal{P}}
\newcommand{\Verifier}{\mathcal{V}}
\newcommand{\MalVerifier}{\Verifier^*}

\newcommand{\Instance}{x}
\newcommand{\Witness}{w}
\newcommand{\Proof}{\pi}

\newcommand{\QueryBound}{q^*}
\newcommand{\MalProof}{\pi^*}

\newcommand{\Simulator}{\mathsf{Sim}}
\newcommand{\View}{\mathsf{View}}

\newcommand{\QueryOut}{\ell_{\text{out}}}

\newcommand{\Randomness}{\mu}

\newcommand{\QueryAnsSet}{T}
\newcommand{\Query}{\alpha}

\newcommand{\Degree}{d}
\newcommand{\NumVars}{m}
\newcommand{\Subcube}{H}

\newcommand{\SCPoly}{F}
\newcommand{\Polys}[3]{#1^{\leq #2}[X_{1},\ldots,X_{#3}]}
\newcommand{\RandPoly}{Q}

\newcommand{\RandLDPoly}{T}

\newcommand{\LagrangePoly}[1]{L_{#1}}

\newcommand{\Language}{\mathcal{L}}
\newcommand{\Relation}{\mathcal{R}}

\newcommand{\OracleTable}{\QueryAnsSet}

\newcommand{\Oracle}{\mathcal{O}}

\newcommand{\AnswerLength}{a}
\newcommand{\NumQueries}{q}
\newcommand{\DecComplexity}{d}
\newcommand{\SoundErr}{s}
\newcommand{\Robustness}{\rho}
\newcommand{\Ecc}{\mathsf{ECC}}
\newcommand{\NonBooProver}{\Prover_{\AnswerLength}}
\newcommand{\NonBooVerifier}{\Verifier_{\AnswerLength}}
\newcommand{\NonBooProof}{\Proof_{\AnswerLength}}

\newcommand{\EccProof}{\tau}

\newcommand{\Rout}{r_{\text{out}}}
\newcommand{\Rin}{r_{\text{in}}}
\newcommand{\Dout}{d_{\text{out}}}
\newcommand{\Din}{d_{\text{in}}}
\newcommand{\Qin}{q_{\text{in}}}
\newcommand{\Qout}{q_{\text{out}}}
\newcommand{\Eout}{\varepsilon_{\text{out}}}
\newcommand{\Ein}{\varepsilon_{\text{in}}}
\newcommand{\RobOut}{\rho_{\text{out}}}
\newcommand{\ProxIn}{\delta_{\text{in}}}
\newcommand{\ProverIn}{\Prover_{\text{in}}}
\newcommand{\ProverOut}{\Prover_{\text{out}}}
\newcommand{\VerifierIn}{\Verifier_{\text{in}}}
\newcommand{\VerifierOut}{\Verifier_{\text{out}}}
\newcommand{\ProverComp}{\Prover_{\text{comp}}}
\newcommand{\VerifierComp}{\Verifier_{\text{comp}}}
\newcommand{\CircVal}{\textsc{Circuit Value}}
\newcommand{\ProofOut}{\Proof_{\text{out}}}

\newcommand{\DecisionCkt}{D}
\newcommand{\DecisionCktOut}{\DecisionCkt_{\text{out}}}
\newcommand{\QuerySet}{I}

\newcommand{\Rand}{r}

\newcommand{\QueryLoc}{\gamma}

\newcommand{\QuerySetOutR}[1]{\QuerySet_{\text{out},#1}}

\newcommand{\ZKProver}{\Prover_{0}}
\newcommand{\ZKVerifier}{\Verifier_0}
\newcommand{\ZKProof}{\Proof_0}

\newcommand{\ZKSimulator}{\overline{\Simulator}}
\newcommand{\Adversary}{\mathcal{A}}

\newcommand{\MaskPoly}{R}
\newcommand{\ZeroP}{Z}
\newcommand{\VerRandVec}{\vec{c}}
\newcommand{\VerRand}{c}
\newcommand{\SumcheckInput}{F}
\newcommand{\ProverFunc}{g}
\newcommand{\ProverPoly}{\hat{\ProverFunc}}

\newcommand{\BadF}{\Tilde{F}}
\newcommand{\SumcheckTotal}{\gamma}
\newcommand{\Event}{E}

\newcommand{\Proximity}{\delta}

\newcommand{\SumDist}{\delta_{\Sigma}}
\newcommand{\RMDist}{\delta_{\text{RM}}}

\newcommand{\ProverLocality}{\ell}
\newcommand{\LocalComp}{f}

\newcommand{\ProofLD}{\Proof_{\hat{P}}}
\newcommand{\HonView}{\nu}
\newcommand{\LDView}{\hat{\nu}}
\newcommand{\SCView}{\nu_{\Sigma}}

\newcommand{\RobSC}{\mathrm{RSC}}
\newcommand{\LDTProx}{\varepsilon_P}

\newcommand{\VQueryAlg}{Q}
\newcommand{\VAcc}{\mathrm{Acc}}

\newcommand{\OSAT}{\mathsf{Oracle}\textsf{-}\mathsf{3SAT}}
\newcommand{\NTIME}{\mathsf{NTIME}}
\newcommand{\ZKVars}{k}
\newcommand{\SCLang}{\mathsf{Sum}}

\newcommand{\tdash}{\text{-}}
\newcommand{\NEXP}{\mathsf{NEXP}}
\newcommand{\NP}{\mathsf{NP}}

\newcommand{\SharpP}{\#\mathsf{P}}
\newcommand{\PZK}{\mathsf{PZK}}
\newcommand{\SZK}{\mathsf{SZK}}

\newcommand{\PCP}{\mathsf{PCP}}
\newcommand{\PZKPCP}{\PZK\tdash\PCP}
\newcommand{\RobustPZKPCP}{\mathsf{r}\PZKPCP}

\newcommand{\Poly}{\mathrm{poly}}

\newcommand{\eqdef}{\ {:=} \ }
\newcommand{\Alg}{\mathcal{A}}
\newcommand{\pST}{\; \middle| \;}

\newcommand{\Circuit}{C}

\newcommand{\defemph}[1]{\textbf{\emph{#1}}}

\newcommand{\jack}[1]{\dtcolornote[Jack]{cyan}{#1}}

\newcommand{\donewpage}{\ifconf\else\newpage\fi}

\title{A Zero-Knowledge PCP Theorem}
\ifconf
\author{Anonymized for submission}
\date{}
\else
\author{Tom Gur \thanks{University of Cambridge. Email: \texttt{tom.gur@cl.cam.ac.uk}. Supported by ERC Starting Grant 101163189 and UKRI Future Leaders Fellowship MR/X023583/1.}
\and Jack O'Connor\thanks{University of Cambridge. Email: \texttt{jack.oconnor@cl.cam.ac.uk}.}
\and Nicholas Spooner\thanks{Cornell University. Email: \texttt{nspooner@cornell.edu}.}}
\date{\today}
\date{}
\fi

\begin{document}

\maketitle
\begin{abstract}
    \noindent We show that for every polynomial $q^*$ there exist polynomial-size, constant-query, non-adaptive PCPs for $\NP$ which are perfect zero knowledge against (adaptive) adversaries making at most $q^*$ queries to the proof. In addition, we construct exponential-size constant-query PCPs for $\NEXP$ with perfect zero knowledge against any polynomial-time adversary. This improves upon both a recent construction of perfect zero-knowledge PCPs for $\SharpP$ (STOC 2024) and the seminal work of Kilian, Petrank and Tardos (STOC 1997).
\end{abstract}

\donewpage
\ifconf\else
\tableofcontents
\fi
\donewpage

\section{Introduction}
\label{sec:intro}
The PCP theorem \cite{AroraS98,ALMSS92} states that for any language in $\NP$, there exists a polynomial-size proof that can be checked probabilistically by reading only a constant number of bits from the proof; or, succinctly,
\begin{equation*}
    \NP \subseteq \PCP[\log n,1] \;,    
\end{equation*}
where $\PCP[r,q]$ is the class of all languages that admit a PCP verifier that uses $O(r)$ random bits and reads $O(q)$ bits of the proof (note that logarithmic randomness complexity implies polynomial proof length).

While PCPs originated from the study of zero-knowledge proofs \cite{GoldwasserMR89}, there seems to be an intrinsic tension between the two notions: PCPs achieve locality by encoding NP witnesses in a manner that \emph{spreads} global information throughout the proof, whereas zero-knowledge proofs aim to \emph{hide} all information except for the validity of the statement. Moreover, PCPs are fundamentally non-interactive objects, whereas interaction is often crucial for zero knowledge.

 Indeed, one must take care in even \emph{defining} zero knowledge PCPs (ZK-PCPs), as it is impossible to achieve non-trivial zero knowledge against a malicious verifier that reads the entire proof. In their seminal work on ZK-PCPs, Kilian, Petrank and Tardos \cite{KilianPT97} identify two regimes of interest:
\begin{enumerate}[label=(\alph*)]
	\item \label{item:poly-regime} \emph{Polynomial-size} PCPs that are zero-knowledge against a verifier that makes at most a \emph{fixed} polynomial number of queries $\QueryBound$ to the PCP. In this regime, they construct ZK-PCPs for NP with polylogarithmic query complexity. We refer to $\QueryBound$ as the ``query bound''.
	\item \label{item:exp-regime} \emph{Exponential-size} PCPs that are zero-knowledge against \emph{any polynomial-time} verifier. In this regime, they construct ZK-PCPs for NEXP with polynomial query complexity.
\end{enumerate}
However, these constructions fall short of a ``zero-knowledge PCP theorem'' in a fundamental way: the query complexity is polylogarithmic, as opposed to $O(1)$, which is a characteristic property of the PCP theorem.

A major obstacle to achieving $O(1)$ query complexity via the \cite{KilianPT97} approach is that the technique used to obtain zero knowledge leads to an inherently \emph{adaptive} honest verifier (i.e., which makes multiple rounds of queries to the proof). Known query-reduction methods apply only to non-adaptive PCPs. Aside from its theoretical interest, non-adaptivity is crucial for some applications \cite{IshaiWY16}. Finally, these constructions only achieve statistical zero knowledge (SZK-PCP) and not perfect zero knowledge (PZK-PCP).

Building on techniques developed in \cite{BenSassonCFGRS17,ChiesaFGS18,ChenCGOS23}, a recent work \cite{GurOS2024} constructed exponential-size PZK-PCPs with polynomially many non-adaptive queries for $\SharpP$. Our first result shows that such PZK-PCPs exist for $\NEXP$, with constant query complexity.
\begin{itheorem}
\label{ithm:NEXP}
    There exist PCPs for $\NEXP$ of exponential length with a non-adaptive verifier that reads $O(1)$ bits of the proof, which are perfect zero-knowledge against any efficient adversary.
\end{itheorem}

Our second result ``scales down'' the above to obtain polynomial-size non-adaptive PZK-PCPs for $\NP$ with constant query complexity.\footnote{For our definition of the class $\PZKPCP$, see \cref{def:PZK-PCP-class}.}
\begin{itheorem}[``Zero-knowledge PCP theorem'']
\label{ithm:NP}
	For any $\QueryBound \leq 2^{\Poly(n)}$, there exist PCPs for $\NP$ of length $\Poly(\QueryBound,n)$ with a non-adaptive verifier that reads $O(1)$ bits of the proof, which are perfect zero-knowledge against any adversary reading at most $\QueryBound$ bits of the proof; i.e.,
	\begin{equation*}
    \NP \subseteq \PZKPCP[\log n,1] \;.    
\end{equation*}
\end{itheorem}

\subsection{Techniques}
\label{sec:tech}

\parhead{PZK-PCPs for nondeterministic computation}
This paper builds on the prior work of \cite{GurOS2024}, which proves a weaker version of \cref{ithm:NEXP} that only captures (decision) $\SharpP$, by constructing non-adaptive PZK-PCPs for the sumcheck problem.

Readers familiar with the PCP literature may wonder why this construction does not lead immediately to a ZK-PCP for NEXP; indeed, the first construction of a PCP for NEXP is essentially a reduction to sumcheck \cite{BabaiFLS91}. However, as we discuss below, even with a PZK-PCP for the sumcheck problem, the BFLS construction is not zero knowledge.

We start by briefly reviewing the BFLS construction, which is a PCP for the $\NEXP$-complete problem $\OSAT$, defined as follows.
\begin{idefinition}[Oracle 3-SAT]
	Let $B \colon \{0,1\}^{r + 3s + 3} \to \{0,1\}$ be a 3-CNF. We say that $B$ is \emph{implicitly satisfiable} if there exists $A \in \Bits^s \to \Bits$ such that for all $z \in \Bits^r, b_1,b_2,b_3 \in \Bits^s$, \\	 $B(z,b_1,b_2,b_3,A(b_1),A(b_2),A(b_3)) = 1$. Let $\OSAT$ be the language of implicitly satisfiable 3-CNFs.
\end{idefinition}
The BFLS PCP consists of two parts. First, the \emph{multilinear extension} of the witness $A$; i.e., a multilinear polynomial $\hat{A}$ over some finite field $\Field$ such that $\hat{A}(x) = A(x)$ for all $x \in \Bits^s$. Second, a sumcheck PCP\footnote{Strictly speaking, a PCP \emph{of proximity}; we will ignore this distinction for this overview.} for the following claim\footnote{To mitigate some technical issues with soundness, the actual construction uses a slightly different summand polynomial.}:
	\begin{equation}
		\label{eq:intro-osat}
		\sum_{\substack{z \in \{0,1\}^r \\ b_1,b_2,b_3 \in \{0,1\}^s}} \hat{B}(z,b_1,b_2,b_3,\hat{A}(b_1),\hat{A}(b_2),\hat{A}(b_3)) = 2^{r+3s} ~,
	\end{equation}
where $\hat{B}$ is an arithmetisation of the circuit $B$.
	
A natural first step is to use the zero-knowledge sumcheck PCP of \cite{GurOS2024} in place of the standard sumcheck PCP, which hides hard-to-compute information about partial sums. However, this \emph{does not} yet yield a ZK-PCP, because the first part of the construction contains an encoding of the witness $A$, which violates the zero-knowledge condition.
 
 To hide $A$, we would like to use the \emph{sumcheck commitment scheme}, introduced by \cite{ChiesaFGS22} (see also \cite{ChiesaFS17}) for their construction of an \emph{interactive} PCP (IPCP) for $\NEXP$. A sumcheck commitment to the polynomial $\hat{A}(\vec X)$ is a random polynomial $C(\vec X,\vec Y)$ of individual degree $d' \geq 2$ in each $Y$-variable such that
	\begin{equation*}
		\sum_{c \in \{0,1\}^\ZKVars} C(\alpha,c) = \hat{A}(\alpha)
	\end{equation*}
	for all $\alpha \in \Field^s$. It was shown in \cite{ChiesaFGS22} that this commitment perfectly hides $\hat{A}$ against all adversaries that make fewer than $2^k$ queries to $C$.
	
	The construction of \cite{BabaiFLS91} uses the sumcheck protocol to reduce checking \eqref{eq:intro-osat} to three random queries to $\hat{A}$. Building on this idea, the ZK-IPCP construction of \cite{ChiesaFGS22} uses the \emph{interactive} sumcheck protocol to open the sumcheck commitment to $\hat{A}$ at those points. Finally, to ensure that those three evaluations do not leak information about the witness, $\hat{A}$ is chosen to be a random multiquartic---rather than the unique multilinear---extension of $A$.
	
    However, the strategy above strongly relies on the \emph{interactivity} of the ZK-IPCP. Indeed, observe that if we were to ``unroll'' this interaction into a PCP, we would simply write down all of $\hat{A}$! We must therefore establish \eqref{eq:intro-osat} \emph{without} opening the sumcheck commitment. That is, we would like to directly check:
    \begin{equation*}
		\sum_{\substack{z \in \{0,1\}^r \\ b_1,b_2,b_3 \in \{0,1\}^s}} \hat{B}\left(z,b_1,b_2,b_3,\sum_{c \in \{0,1\}^k} C(b_1,c),\sum_{c \in \{0,1\}^k} C(b_2,c),\sum_{c \in \{0,1\}^k} C(b_3,c)\right) = 2^{r+3s} ~,
	\end{equation*}
	To do this, we first ``pull out'' the three inner summations. There are various ways this can be achieved; we follow a linearisation approach. First observe that (assuming $\hat{A}$  takes boolean values on $\{0,1\}^s$), the LHS of \eqref{eq:intro-osat} is equal to
	\begin{equation*}
		\sum_{a_1,a_2,a_3 \in \{0,1\}} \sum_{\substack{z \in \{0,1\}^r \\ b_1,b_2,b_3 \in \{0,1\}^s}} \hat{B}(z,b_1,b_2,b_3,a_1,a_2,a_3) \cdot \prod_{i=1}^3 (\hat{A}(b_i) - (1-a_i))~,
	\end{equation*}
	since the product expression ``zeroes out'' any term of the sum for which some $a_i \neq \hat{A}(b_i)$.
	Next, observe that for any $b_1,b_2,b_3,a_1,a_2,a_3$, provided that $\Field$ has characteristic different from $2$,
	\begin{equation*}
		\prod_{i=1}^3 (\hat{A}(b_i) - (1-a_i))
		 = \prod_{i=1}^3 \sum_{c \in \{0,1\}^k}  \left(C(b_i,c) - \frac{1-a_i}{2^k}\right) 
		 = \sum_{c_1,c_2,c_3 \in \{0,1\}^k} \prod_{i=1}^3 \left(C(b_i,c_i) - \frac{1-a_i}{2^k}\right)~.
	\end{equation*}
	Taken together, we see that checking \eqref{eq:intro-osat} is equivalent to checking
	\begin{equation*}
		\sum_{c_1,c_2,c_3 \in \{0,1\}^k} \sum_{a_1,a_2,a_3 \in \{0,1\}} \sum_{\substack{z \in \{0,1\}^r \\ b_1,b_2,b_3 \in \{0,1\}^s}} \hat{B}(z,b_1,b_2,b_3,a_1,a_2,a_3) \prod_{i=1}^3 \left(C(b_i,c_i) - \frac{1-a_i}{2^k}\right) = 2^{r+3s}~,
	\end{equation*}
	which can be proven in zero knowledge using the \cite{GurOS2024} ZK-PCPP for sumcheck. In particular, the PCPP simulator requires only polynomially many evaluations of the summand to simulate polynomially many queries to the proof. By setting $k = \omega(\log n)$ we can simulate those evaluations by lazily simulating $C$ as a uniformly random polynomial (via an algorithm of \cite{BenSassonCFGRS17}), and then evaluating the summand directly. The query complexity of the verifier is $\Poly(n)$.
 
	Our result for $\NP$ is obtained in a similar way, scaling the parameters appropriately. In particular, for a given adversary query bound $\QueryBound$, we can set $k = O(\log q^*)$. The query complexity of the verifier is then $\Poly(\log n, \log q^*)$. We also point out that the honest prover in our construction for $\NP$ is efficient, given a valid witness as input.

\parhead{Proof composition and zero knowledge} In order to obtain constant query ZK-PCPs for $\NP$, we would like to apply the proof composition paradigm \cite{BenSassonGHSV06} to our ZK-PCPs. This involves composing a robust outer PCP with an inner PCP of proximity to obtain a PCP which inherits the randomness complexity of the former and query complexity of the latter. To do this, we first need to strengthen our ZK-PCPs to satisfy \emph{robust} soundness (i.e., the local view of the verifier must be far from an accepting view, with high probability). Then we show that proof composition \emph{preserves} the zero knowledge of the outer PCP. To the best of our knowledge, this is the first composition theorem for ZK-PCPs.

Our first step is to obtain robust ZK-PCPs for $\NP$ with polylogarithmic query complexity. While the PZK-PCP of \cite{GurOS2024} (upon which our PZK-PCP builds) is an algebraic construction, it does \emph{not} have constant robust soundness as written. We present a modification of the \cite{GurOS2024} PCP for sumcheck, following the ``query bundling'' approach of \cite{BenSassonGHSV06}, that has constant robust soundness.

The key challenge is to show that this modification preserves zero knowledge. In fact, we will show that a much more general class of ``local'' transformations preserve zero knowledge. This class includes not only query bundling, but also the subsequent steps of alphabet reduction and proof composition. To capture this class formally, we define a new notion, \emph{locally computable proofs}.

\begin{idefinition}[Locally computable algorithms (informal; see \cref{def:loc-proofs})]
     Let $\Adversary$ and $\Adversary_0$ be randomized algorithms. We say that $\Adversary$ is \defemph{$\ell$-locally computable} from $\Adversary_0$ if there exists an efficient, deterministic oracle algorithm $\LocalComp$ making at most $\ProverLocality$ queries to its oracle such that, for every input $x$, the following two distributions are identically distributed:
    \begin{equation*}
        \Adversary(x), \quad (f^{\Proof_0} ~|~ \Proof_0 \gets \Adversary_0(x)).
    \end{equation*}
\end{idefinition}

We show that if the PCP prover algorithm $\Prover$ is $\ell$-locally computable from a zero-knowledge PCP prover $\Prover_0$, then provided $\ell$ is asymptotically smaller than the query bound on the ZK-PCP, $\Prover$ inherits the zero knowledge guarantee of $\Prover_0$. Intuitively, this is true because if a proof $\Proof$ is locally computable from a proof $\Proof_0$, and $\Proof_0$ is zero knowledge, then we can apply $f$ to the simulator for $\Proof_0$ to obtain a simulator for $\Proof$. This notion is surprisingly versatile, and allows us to prove zero knowledge in an array of distinct settings:

\begin{itemize}
    \item \parhead{Robustification of \cite{GurOS2024}} We show that our modified \cite{GurOS2024} construction is locally computable from the original construction, and thus inherits zero knowledge.
    
    \item \parhead{Alphabet reduction} Recall that alphabet reduction allows us to transform a robust PCP over a large alphabet into a boolean PCP, while maintaining robustness. It is performed by encoding each symbol of the PCP with a good error correcting code. More formally, if $\Proof$ is a distribution over robust PCP proofs over the alphabet $\Bits^a$ for some $a \in \N$, and $\Ecc\colon \Bits^a \to \Bits^b$ is a systematic binary error-correcting code of constant relative distance and rate, then we can obtain a robust boolean PCP by defining a new proof $\tau(\alpha) \eqdef \Ecc(\Proof(\alpha))$ for every proof index $\alpha$, and writing $\tau$ over the alphabet $\Bits$. Then $\tau$ is $1$-locally computable from $\Proof$ by the following function:
    \begin{equation*}
        f^{\Proof}(\Query, i) = \Ecc(\Proof(\Query))_i,
    \end{equation*}
    where $i \in [b]$. We note that prior work on alphabet reduction for ZKPCPs \cite{HazayVW22} required a much more complex construction because their ZKPCP achieves only a quadratic gap between the honest verifier's query complexity and the query bound.

    \item \parhead{Proof composition} Recall that composition of an outer PCP system $(\ProverOut, \VerifierOut)$ for $\Language$ with an inner PCP of proximity $(\ProverIn, \VerifierIn)$ for circuit evaluation proceeds as follows. Let $\ProofOut \gets \ProverOut$. For every choice $r \in \Bits^{\Rout}$ of $\VerifierOut$'s randomness, the composed prover computes $\VerifierOut$'s query set $Q(r)$ and an ``inner proof'' $\Proof_r \eqdef \ProverIn(\VerifierOut, \ProofOut|_{Q(r)})$, which attests to the fact that if $\VerifierOut$ performed its verification of $\ProofOut$ using randomness $r$, then it would have accepted. Each inner proof is a function of at most $\QueryOut$ many locations of $\ProofOut$, where $\QueryOut$ is the query complexity of $\VerifierOut$. The composed proof is given by $(\ProofOut, (\Proof_r)_{r \in \Bits^{\Rout}})$. 

    Hence the composed proof is $\QueryOut$-locally computable from $\ProofOut$, by the following function: 
    \begin{equation*}
        f^{\ProofOut}(\Oracle, i) = \begin{cases}
            \ProofOut(i) ~&\text{if}~\Oracle = \ProofOut
            \\
            \ProverIn(\VerifierOut, \ProofOut|_{Q(r)})_i ~&\text{if}~\Oracle = \Proof_r,~\text{for some}~ r\in\Bits^{\Rout}
        \end{cases}
        .
    \end{equation*}
    Note that the composed PCP is locally computable from the outer PCP, so only the outer PCP needs to be zero knowledge to ensure the composed PCP is zero knowledge. Therefore we can employ the existing (non-ZK) PCP of proximity for circuit evaluation of \cite{BenSassonGHSV06} as the inner PCPP.
\end{itemize}

\subsection{Open problems}
This work shows that for any polynomial $\QueryBound$, any language in NP has a polynomial-sized proof that can be probabilistically checked by probing only $O(1)$ bits, but where any set of $\QueryBound$ bits carries no information about the witness. This can be viewed a zero-knowledge PCP theorem that matches parameters of the original PCP theorem \cite{AroraS98,ALMSS92}. Since then, stronger versions of the PCP theorem have been shown. It is tempting to ask whether zero-knowledge PCPs can match the strongest constructions of standard PCPs.

In particular, one of the most immediate open questions is whether it is possible to obtain ZK-PCPs with nearly-linear length. Optimising the proof length of PCPs received much attention for decades after the first proof of the PCP theorem, where the current state of the art achieves quasilinear proof length \cite{BS08,Din07}. We ask whether the same can be obtained for zero-knowledge PCPs.
\begin{openproblem}[Short ZK-PCPs]
    Do there exist $O(1)$-query ZK-PCPs for $\NP$ with proof length $\tilde{O}(n)$?
\end{openproblem}
We remark that our algebraic zero-knowledge techniques are based on Reed-Muller arithmetisation and the sumcheck protocol, whereas (non-ZK) constructions of quasilinear length PCPs are based on Reed-Solomon arithmetisation and combinatorial gap amplification, hence new ideas are necessary for such strengthening of our theorem.

Even more ambitiously, one could ask whether it is possible to transform any construction of a (non-ZK) PCP to a ZK-PCP while preserving its parameters. 
\begin{openproblem}[PCP to ZK-PCP transformation]
    Is there a black-box transformation that imbues a PCP construction with zero knowledge?
\end{openproblem}
We note that prior to \cite{GurOS2024}, all works on ZK-PCPs (see below) followed this approach, but those transformations either introduce adaptivity or only achieve weak ZK guarantees, and none preserve query complexity. Similiar transformations are known, e.g., for multi-prover interactive proofs \cite{BenOrGKW88} and their quantum analogues \cite{GriloSY19,MastelS24} whereas in other models, such as in zero-knowledge streaming interactive proofs \cite{cormode2023streaming}, we have a zero knowledge sumcheck protocol but no generic transformation is known.

We remark that our techniques make whitebox use of the structure of the \cite{BabaiFLS91} PCP, and rely strongly on the \cite{GurOS2024} ZK-PCP for sumcheck; hence, obtaining a generic transformation would require new ideas.  

\parhead{The quantum PCP conjecture.} Finally, we highlight a connection between zero-knowledge PCPs and one of the most imporant open problems in quantum complexity theory: the quantum PCP (QPCP) conjecture.

Most classical constructions of PCPs rely on an encoding of the $\NP$ witness via a locally-testable and (relaxed) locally-decodable code \cite{BenSassonGHSV06,gur2020relaxed}. Even though there is growing evidence that good quantum LTCs may exist \cite{aharonov2015quantum,leverrier2022towards,anshu2023nlts,dinur2024expansion}, quantum codes \emph{cannot} be locally decodable due to the no-cloning theorem. This is one of the main barriers towards applying algebraic and coding-theoretic techniques to QPCPs.

Quantum codes are fundamentally tied to zero knowledge. It is a well-known fact that the erasure of a subset of qubits of a codeword is correctable if and only if the reduced density matrix on that subset is independent of the encoded state. In other words, roughly speaking, a quantum code has good distance if and only if it satisfies a quantum analogue of the PZK property for PCPs. Thus PZK-PCPs are perhaps the closest classical analogue of QPCPs, and studying them may help to shed light on the QPCP conjecture.

\subsection{Related work}
The first zero-knowledge PCPs appeared in the work of Kilian, Petrank and Tardos \cite{KilianPT97}. Later works \cite{IshaiMS12,IshaiSVW13,IshaiW14, IshaiMSX15} simplified this construction, and extended it to PCPs of proximity and the closely related notion of zero-knowledge locally testable codes (LTCs). These constructions rely on an adaptive honest verifier, and hence it is unclear how to use proof composition to improve their parameters. \cite{IshaiMS12,IshaiMSX15} showed that PCPs which are zero knowledge against any efficient adversary and where the proof oracle is described by a polynomial-sized circuit exist only for languages in $\SZK$.

Another line of work, motivated by cryptographic applications, focuses on obtaining SZK-PCPs for NP with a \emph{non-adaptive} honest verifier from leakage resilience. These results come with caveats, achieving either a weaker notion of zero knowledge known as \emph{witness indistinguishability} \cite{IshaiWY16}, or simulation against adversaries making only quadratically many more queries than the honest verifier \cite{HazayVW22}. See \cite{Weiss22} for a survey of this line of work.

In related models that allow for interaction, zero-knowledge proofs are easier to construct. We know that $\mathsf{PZK\text{-}MIP} = \mathsf{MIP}$ ($= \NEXP$) \cite{BenOrGKW88}, where $\mathsf{MIP}$ is the class of languages with a multi-prover interactive proofs. The quantum analogue of this result, $\mathsf{PZK\text{-}MIP}^* = \mathsf{MIP}^*$ ($= \mathsf{RE}$), is also known to hold \cite{ChiesaFGS22,GriloSY19,MastelS24}. The constructions in this work draw inspiration from a similar result for interactive PCPs (IPCPs) \cite{KalaiR08}, an interactive generalisation of PCPs (and special case of IOPs): $\mathsf{PZK\text{-}IPCP} = \mathsf{IPCP} = \mathsf{NEXP}$ \cite{BenSassonCFGRS17,ChiesaFS17}.

\donewpage

\section{Preliminaries}
\label{sec:prelims}
Throughout $\Field$ is a finite field. $\CircVal$ denotes the $\mathsf{P}$-complete language $\{(\Circuit, x) \in \Bits^*: \Circuit~\text{is an encoding of a boolean circuit}$ such that $\Circuit(x) = 1\}$. For an oracle $\Proof$, we use $\Domain(\Proof)$ to denote the domain of $\Proof$. For $t \in \N$, and oracles $\Proof_1,\dots,\Proof_t$, we denote their concatenation by $(\Proof_1,\dots,\Proof_t)$, and 
use the notation $(\Proof_i, q)$ to denote the index $q \in \Domain(\Proof_i)$. For a pair language $\Language$, for every string $x \in \Bits^n$, we denote the set $\Language[\Instance] \eqdef \{y : (\Instance,y) \in \Language \}$. For a relation $\Relation$, we denote the language corresponding to $\Relation$ by $\Language(\Relation) \eqdef \{\Instance : (\Instance, \Witness) \in \Relation \}$.

\parhead{Algorithms} We write $\mathcal{A}^{\Proof}(\Instance)$ to denote the output of $\mathcal{A}$ when given input $\Instance$ (explicitly) and oracle access to $\Proof$. We use the abbreviation PPT to denote probabilistic polynomial-time. If an oracle algorithm makes at most $\QueryBound(n)$ oracle queries we say that it is $\QueryBound$-query-bounded. We generally omit the internal randomness of an algorithm from probability statements; that is, we write $\Pr[\Alg(x) = 0]$ to mean $\Pr_{r \gets \Bits^n}[\Alg(x; r) = 0]$.

\parhead{Polynomials} For $\Degree \in \N$, we write $\Polys{\Field}{\Degree}{\NumVars}$ for the ring of polynomials in $\NumVars$ variables over $\Field$ of \emph{total} degree at most $\Degree$. For $\vec{\Degree} \eqdef (\Degree_1,\dots, \Degree_\NumVars) \in \N^\NumVars$, we will write $\Polys{\Field}{\vec{\Degree}}{\NumVars}$ for the ring of polynomials in $\NumVars$ variables over $\Field$ of \emph{individual} degree $\Degree_i$ in the $i$-th variable. For a product set $S = S_1 \times \cdots \times S_\NumVars \subseteq \Field^\NumVars$, $\vec a \in S$, and $\vec \Degree = (|S_1|-1,\ldots,|S_\NumVars|-1)$, we denote by $\LagrangePoly{S,\vec a}$ the unique element of $\Polys{\Field}{\vec{\Degree}}{\NumVars}$ such that for all $\vec b \in S$,
\[
	\LagrangePoly{S,\vec a}(\vec b) = \begin{cases}
 			1 & \text{if $\vec b = \vec a$, and} \\
 			0 & \text{otherwise.}
 		\end{cases}
\]

\parhead{Vector reversal and truncation} For a vector $\vec{\alpha} = (\alpha_1, \dots, \alpha_n) \in \Field^n$, we denote by $\revvec{\alpha}$ the reversed vector $(\alpha_n, \alpha_{n-1}, \dots, \alpha_1)$. For any $i \in [n]$, we denote the truncation of $\vec{\alpha}$ to its first $i$ entries by $\vec{\alpha}_i \eqdef (\alpha_1, \dots, \alpha_i)$.

\subsection{Coding theory}
Let $\Alphabet$ be an alphabet and let $\BlockLength \in \N$. A \emph{code} $\Code$ is a subset $\Code \subseteq \Alphabet^\BlockLength$. Given two strings $x, y \in \Alphabet^n$, we denote the relative Hamming distance between $x$ and $y$ by $\HammingDist(x, y) \eqdef \left|\{i \in [n]:x_i \neq y_i\}\right|/n$. We say that a vector $x$ is $\varepsilon$-far from a set $S \subseteq \Alphabet^n$ if $\min_{y \in S} \HammingDist(x,y) \geq \varepsilon$.

\parhead{Reed--Muller codes} The Reed--Muller (RM) code is the code consisting of evaluations of multivariate low-degree polynomials over a finite field. Given a finite field $\Field$, and positive integers $\NumVars$ and $\Degree$ we denote by $\ReedMuller[\Field, \NumVars,\Degree]$ the linear code consisting of evaluations of $\NumVars$-variate polynomials over $\Field$ of total degree at most $\Degree$. 

\subsection{PCPs}

\begin{definition}[PCP]
    A \defemph{probabilistically checkable proof (PCP)} for a relation $\Relation$ consists of a prover $\Prover$ and a PPT verifier $\Verifier$ such that the following holds.
    \begin{enumerate}
        \item \parhead{Completeness} For every $(\Instance,\Witness)\in \Relation$,
        \begin{equation*}
            \Pr_{\Proof \gets \Prover(\Instance, \Witness)}[\Verifier^{\Proof}(\Instance) = 1] = 1.
        \end{equation*}

        \item \parhead{Soundness} For every $\Instance \notin \Language(\Relation)$ and every oracle $\MalProof$,
        \begin{equation*}
            \Pr[\Verifier^{\MalProof}(\Instance) = 1] \leq \frac{1}{2}.
        \end{equation*}
    \end{enumerate}
    We say that a PCP is \defemph{efficient} if $\Prover(\Instance, \Witness)$ can be computed efficiently for any $(\Instance, \Witness) \in \Relation$. 
\end{definition}
We formally define PCPs for relations (rather than langauges) to allow us to discuss efficient PCP provers in the $\NP$ setting.

\begin{definition}[PCPP]
    For $\Proximity\colon \N \to [0,1]$, a \defemph{probabilistically checkable proof of proximity (PCPP)} for a pair language $\Language$ with proximity parameter $\Proximity$ consists of a prover $\Prover$ and a verifier $\Verifier$ such that the following holds for every pair of strings $(x, y)$.
    \begin{enumerate}
        \item \parhead{Completeness} If  $(\Instance,y) \in \Language$, 
        \begin{equation*}
            \Pr_{\Proof \gets \Prover(\Instance, y)}[\Verifier^{(y, \Proof)}(\Instance) = 1] = 1.
        \end{equation*}

        \item \parhead{Soundness} If $y$ is $\delta$-far from the set $\Language[\Instance] \eqdef \{y : (\Instance,y) \in \Language \}$, then for every oracle $\MalProof$,
        \begin{equation*}
            \Pr[\Verifier^{(y, \MalProof)}(\Instance) = 1] \leq \frac{1}{2}.
        \end{equation*}
    \end{enumerate}
\end{definition}

\begin{definition}[Non-adaptive PCP verifiers]
    A \defemph{non-adaptive PCP verifier} is an algorithm of the form $\Verifier^\Proof(\Instance; \Randomness) = \DecisionCkt (\Instance, \Proof|_{\VQueryAlg(\Instance; \Randomness)};\Randomness)$, where $\DecisionCkt \colon \Bits^n \times \Alphabet^\NumQueries \times \Bits^{\Rand(n)} \to \Bits$ is the \defemph{decision algorithm}, and $\VQueryAlg\colon \Bits^n \times \Bits^{r(n)} \to \Domain(\Proof)^\NumQueries$ is the \defemph{query algorithm}.
\end{definition}

\begin{definition}[Accepting view of a PCP verifier]
    For non-adaptive $\NumQueries$-query PCP verifiers of the form $\Verifier^\Proof(\Instance; \Randomness) = \DecisionCkt (\Instance, \Proof|_{\VQueryAlg(\Instance; \Randomness)};\Randomness)$, where $\DecisionCkt \colon \Bits^n\times\Alphabet^\NumQueries \times \Bits^{\Rand(n)} \to \Bits$, and $\VQueryAlg\colon \Bits^n \times \Bits^{r(n)} \to \Domain(\Proof)^\NumQueries$, we define, for each choice of randomness $\Randomness \in \Bits^{\Rand(n)}$ the \defemph{set of accepting views of $\Verifier$} to be:
    \begin{equation*}
        \VAcc(\Verifier(x; \Randomness)) \eqdef \{a \in \Alphabet^\NumQueries : \DecisionCkt(x, a; \Randomness) = 1\}.
    \end{equation*}
\end{definition}
\begin{definition}[Accepting view of a PCP of proximity verifier]
    For non-adaptive $\NumQueries$-query PCP of proximity verifiers of the form $\Verifier^{(y,\Proof)}(\Instance; \Randomness) = \DecisionCkt (x, (y ,\Proof)|_{\VQueryAlg(\Instance; \Randomness)};\Randomness)$, where $\DecisionCkt \colon \Bits^n \times \Alphabet^\NumQueries \times \Bits^{\Rand(n)} \to \Bits$, and $\VQueryAlg\colon \Bits^n \times \Bits^{r(n)} \to \Domain((y,\Proof))^\NumQueries$, we define, for each choice of randomness $\Randomness \in \Bits^{\Rand(n)}$, the \defemph{set of accepting views of $\Verifier$} to be:
    \begin{equation*}
        \VAcc(\Verifier(x; \Randomness)) \eqdef \{a \in \Alphabet^\NumQueries : \DecisionCkt(x, a; \Randomness) = 1\}.
    \end{equation*}
\end{definition}
In both of the above definitions when it is clear from context we will omit $\Randomness$ and write $\VAcc(\Verifier(x))$. 

\begin{definition}[Robust soundness]
    A non-adaptive PCP for a relation $\Relation$ has \defemph{robustness $\Robustness$} if for every $x \notin \Language(\Relation)$, it holds that, for every oracle $\MalProof$,
    \begin{equation*}
      \Pr_{\Randomness}\left[\HammingDist(\MalProof|_{\VQueryAlg(\Instance)}, \VAcc(\Verifier(\Instance; \Randomness))) \leq \Robustness\right] \leq \frac{1}{2}.
    \end{equation*}
    A non-adpative PCP of proximity for a language $\Language$ has \defemph{robustness $\Robustness$} if for every $(x,y)$ such that $y$ is $\delta$-far from $\Language[x]$, it holds that, for every oracle $\MalProof$,
    \begin{equation*}
    \Pr_{\Randomness}\left[\HammingDist((y, \MalProof )|_{\VQueryAlg(\Instance)}, \VAcc(\Verifier(\Instance;\Randomness))) \leq \Robustness\right] \leq \frac{1}{2}.
    \end{equation*}
\end{definition}

\begin{definition}
	For $\Robustness\colon \N \to [0,1]$, a PCP has \defemph{expected robustness $\Robustness$} if for every $x \notin \Language(\Relation)$, we have for every oracle $\MalProof$, 
	\begin{equation*}
		\mathop{\mathbb{E}}_{\Randomness}\left[\HammingDist(\MalProof|_{\VQueryAlg(x)}, \VAcc(\Verifier(x;\Randomness)))\right] \geq \Robustness(|x|).
	\end{equation*}
	A PCP of proximity has \defemph{expected robustness $\Robustness$} if for every $(x,y)$ such that $y$ is $\delta$-far from $\Language[x]$, it holds that, for every oracle $\MalProof$,
    \begin{equation*}
		\mathop{\mathbb{E}}_{\Randomness}\left[\HammingDist((y, \MalProof)|_{\VQueryAlg(x)}, \VAcc(\Verifier(x;\Randomness)))\right] \geq \Robustness(|x|).
	\end{equation*}
\end{definition}

The following proposition relating robustness and expected robustness appears as Proposition 2.10 in \cite{BenSassonGHSV06}.
\begin{proposition}[\cite{BenSassonGHSV06}]
	\label{prop:exp-rob-to-rob}
	If a PCPP has expected robustness $\Robustness$, then for every $\varepsilon \leq \Robustness$, it has robust-soundness error $1-\varepsilon$ with robustness parameter $\Robustness - \varepsilon$. 
\end{proposition}

\subsection{Zero-knowledge PCPs}

\begin{definition}[View]
	For a PCP $(\Prover, \Verifier)$ and a (possibly malicious) verifier $\MalVerifier$, $\View_{\MalVerifier,\Prover}(\Instance,\Witness)$ denotes the view of $\MalVerifier$ with input $\Instance$ and oracle access to $\Proof \gets \Prover(\Instance,\Witness)$. That is, $\View_{\MalVerifier,\Prover}(\Instance,\Witness)$ comprises $\MalVerifier$'s random coins and all answers to $\MalVerifier$'s queries to $\Proof$. 
\end{definition}

\begin{definition}[Perfect zero-knowledge PCP]
    We say that a PCP system $(\Prover, \Verifier)$ for a relation $\Relation$ is \defemph{perfect zero knowledge} (a PZK-PCP) with query bound $\QueryBound$ if for every (possibly malicious) adaptive $\QueryBound$-query-bounded verifier $\MalVerifier$ there exists a PPT algorithm $\Simulator_{\MalVerifier}$ (the simulator), such that for every $(\Instance,\Witness)\in \Relation$, $\Simulator_{\MalVerifier}(\Instance)$ is distributed identically to $\View_{\MalVerifier,\Prover}(\Instance,\Witness)$.
\end{definition}

All of our constructions satisfy the stronger notion of \emph{black-box} zero knowledge, defined next.

\jack{Do we want sim to be poly in $\QueryBound$ here too?}

\begin{definition}[Black-box perfect zero-knowledge PCP]
    We say that a PCP system $(\Prover, \Verifier)$ for a relation $\Relation$ is \defemph{black-box} perfect zero knowledge with query bound $\QueryBound$ if there exists a PPT algorithm $\Simulator$ (the simulator), such that for every (possibly malicious) adaptive $\QueryBound$-query-bounded verifier $\MalVerifier$, and for every $(\Instance,\Witness) \in \Relation$, $\Simulator^{\MalVerifier}(\Instance)$ is distributed identically to $\View_{\MalVerifier,\Prover}(\Instance,\Witness)$.
\end{definition}

We define the complexity class $\PZKPCP$ and its robust version.

\begin{definition}[$\PZKPCP, \RobustPZKPCP$]
    \label{def:PZK-PCP-class}
    For functions $r, q, \QueryBound \colon \N \to \N$, we define the complexity class \defemph{$\PZKPCP_{\QueryBound, \Alphabet}[r(n), q(n)]$} to be the set of all languages that admit perfect zero-knowledge PCPs written over the alphabet $\Alphabet$ with randomness complexity $O(r(n))$, query complexity $O(q(n))$ and query bound $\QueryBound(n)$. Further, we define the complexity class \defemph{$\RobustPZKPCP_{\QueryBound, \Alphabet}[r(n), q(n)]$} to be the set of all languages that admit perfect zero-knowledge PCPs (with the same parameters) with robustness $\Omega(1)$.

    We define the complexity class 
    \begin{equation*}
        \PZKPCP[r, q] \eqdef \bigcap_{\QueryBound} \PZKPCP_{\QueryBound, \Bits}[r + \log \QueryBound, q]
    \end{equation*}
    where the intersection ranges over all functions $q^*(n)$ that are bounded by $2^{p(n)}$ for some polynomial $p$.
\end{definition}
That is, $\Language \in \PZKPCP[r,q]$ if for any query bound $q^*$ (that is at most exponential in $n$), $\Language$ has a perfect zero-knowledge PCP of length $\Poly(2^r,q^*)$ over the boolean alphabet with query complexity $O(q)$. 

We also define zero-knowledge PCPs of proximity.

\begin{definition}[View of PCP of proximity]
	For a PCP of proximity $(\Prover, \Verifier)$ and a (possibly malicious) verifier $\MalVerifier$, $\View_{\MalVerifier,\Prover}(\Instance,y)$ denotes the view of $\MalVerifier$ with explicit input $\Instance$ and oracle access to both implicit input $y$ and $\Proof \gets \Prover(\Instance)$. That is, $\View_{\MalVerifier,\Prover}(\Instance,\Witness)$ comprises $\MalVerifier$'s random coins and all answers to $\MalVerifier$'s queries to $(y,\Proof)$. 
\end{definition}

\begin{definition}[Black-box perfect zero-knowledge PCP of proximity]
    We say that a PCP of proximity system $(\Prover, \Verifier)$ for a pair language $\Language$ is \defemph{black-box} perfect zero knowledge with query bound $\QueryBound$ if there exists an algorithm $\Simulator$, such that for every (possibly malicious) adaptive $\QueryBound$-query-bounded verifier $\MalVerifier$, and for every $(\Instance,y) \in \Language$, $\Simulator^{(\MalVerifier,y)}(\Instance)$ is distributed identically to $\View_{\MalVerifier,\Prover}(\Instance,y)$, and moreover, $\Simulator$ runs in time $\Poly(|x|, \QueryBound)$.
\end{definition}
\subsection{Low-degree testing}

\begin{definition}[Vector-valued Reed--Muller code]
    The vector-valued Reed--Muller code consists of tuples of evaluations of multivariate low-degree polynomials over a finite field. Given a finite field $\Field$, and positive integers $\NumVars$, $\Degree$ and $k$ we denote by $\ReedMuller^k[\Field, \NumVars, \Degree]$ the linear code over the alphabet $\Field^k$ consisting of $k$-tuples of evaluations of $m$-variate polynomials over $\Field$ of total degree at most $\Degree$. 
    More formally, 
    \begin{equation*}
        \ReedMuller^k[\Field, \NumVars, \Degree] \eqdef \{(p_1(\vec{x}), \dots, p_{k}(\vec{x})) \in \Field^k : \vec{x} \in \Field^\NumVars,~ \forall i \in [k], ~p_i(\vec{X}) \in \Polys{\Field}{\Degree}{\NumVars}\}.
    \end{equation*}
\end{definition}

The following theorem appears as Proposition 5.7 in \cite{Paradise2021}.
\begin{theorem}[Robust vector-valued low-degree test \cite{Paradise2021}] 
\label{thm:bundled-ldt}
    For $\delta>0$, provided $|\Field| > 25k$ there exists a test that given oracle access to a function $F\colon \Field^\NumVars\to\Field^k$, makes $|\Field|$ queries to $F$, runs in time $\Poly(|\Field|, \NumVars, \Degree, k)$, and:
    \begin{itemize}
        \item if $F \in \ReedMuller^k[\Field, \NumVars, \Degree]$ then the test accepts with probability $1$;
        \item if $F$ is $\varepsilon$-far from $\ReedMuller^k[\Field, \NumVars,\Degree]$ then the expected distance of the tester's view from any accepting view is at least $\Omega(\varepsilon)$.
    \end{itemize}
\end{theorem}

\section{Zero-knowledge proof composition}
In this section, we will prove that PCP proof composition (as in \cite{BenSassonGHSV06}) preserves zero knowledge. Towards this end, we introduce the notion of \emph{locally computable proofs}, which are proofs in which each symbol can be deterministically and efficiently computed from a local view of another proof. For example, in proof composition, the each inner proof is locally computable from the outer (zero-knowledge) proof.

In \cref{sec:loc-proofs}, we formally define locally computable proofs and show that if a PCP is locally computable from a PZK-PCP, then the locally computable PCP also exhibits perfect zero knowledge. In \cref{sec:alphabet-reduction}, we show how to convert a PZK-PCP over an arbitrary alphabet into a boolean PZK-PCP. In \cref{sec:proof-comp}, we show that proof composition preserves perfect zero knowledge. Specifically, if the outer PCP of the composition is perfect zero knowledge, then the composed PCP will also be perfect zero knowledge.

\subsection{Locally computable proofs}
\label{sec:loc-proofs}

We introduce the notion of $\ProverLocality$-locally computable proofs. This is a new notion that we consider to be of independent interest and will be useful towards showing that query bundling, alphabet reduction and proof composition preserve zero knowledge.

\begin{definition}
\label{def:loc-proofs}
    Let $\Adversary$ and $\Adversary_0$ be randomized algorithms, and let $\ell\colon \N\to\N$. We say that $\Adversary$ is \defemph{$\ell$-locally computable} from $\Adversary_0$ on $C \subseteq \Bits^*$ if there exists an efficient, deterministic oracle algorithm $\LocalComp$ making at most $\ProverLocality$ queries to its oracle such that, for every $x \in C$, the following two distributions are  identical:
    \begin{equation*}
        \Adversary(x); \qquad (f^{\Proof_0} ~|~ \Proof_0 \gets \Adversary_0(x)).
    \end{equation*}
\end{definition}

\begin{lemma}
\label{lem:locally-computable-implies-zk}
    Let $(\ZKProver, \ZKVerifier)$ be a PZK-PCP for a relation $\Relation$ with query bound $\QueryBound$, and let $(\Prover, \Verifier)$ be a PCP for $\Language$ such that $\Prover$ is $\ProverLocality$-locally computable from $\ZKProver$ on $\Language$. Then $(\Prover, \Verifier)$ is perfect zero knowledge with query bound $\QueryBound/\ell$.
\end{lemma}

While we state and prove \cref{lem:locally-computable-implies-zk} for the case of PCPs, the analysis straightforwardly extends to the case of PCPs of proximity.

The intuition behind \cref{lem:locally-computable-implies-zk} is that if $\ZKProof \gets \ZKProver$ is zero knowledge, then since any local view of $\ZKProof$ is efficiently simulatable and $\Proof\gets \Prover$ is computable from a local view of $\ZKProof$, we can combine these two functionalities to obtain a simulator for $\Prover$. More formally, let $\MalVerifier$ be a (possibly adaptive) malicious verifier of $(\Prover, \Verifier)$. Our goal is to construct a simulator for the view of $\MalVerifier$ and our strategy for doing so consists of two parts. In \cref{construc:composed-adversary}, we construct a hybrid simulator $\Adversary_{\MalVerifier}$ for $\MalVerifier$, which is given oracle access to the (zero-knowledge) proof $\ZKProof \gets \ZKProver$. This hybrid simulator uses the local computability of $\Proof \gets \Prover$ and its oracle access to $\ZKProof$ to simulate $\MalVerifier$'s view of $\Proof$. \cref{construc:full-sim} is our simulator for $\MalVerifier$. As $\Adversary_{\MalVerifier}$ is itself a malicious verifier of the proof $\ZKProof$, we can invoke the zero-knowledge property of $\ZKProof$ to simulate the view of $\Adversary_{\MalVerifier}$. Our simulator for $\MalVerifier$ works by running $\Adversary_{\MalVerifier}$ and answering its queries to $\ZKProof$ using the simulator for $\ZKProof$.

\begin{mdframed}[nobreak=true]
    \begin{construction} 
    \label{construc:composed-adversary} 
    A hybrid simulator for any malicious verifier $\MalVerifier$ of the PCP $(\Prover, \Verifier)$, a PCP which is $\ProverLocality$-locally computable from $(\ZKProver,\ZKVerifier)$ by a function $\LocalComp$. Receives oracle access to $\ZKProof\gets\ZKProver$.
    
    \noindent
        $\Adversary_{\MalVerifier}^{\ZKProof}(x, \Rand)$:
        \begin{enumerate}[nolistsep]
            \item Initialise an empty list $\OracleTable$.
            \item Run $\MalVerifier$ on random coins $\Rand$. Every time $\MalVerifier$ makes a query $\Query$ to $\Proof$, compute $\beta \eqdef f^{\ZKProof}(\Query)$, add $(\Query, \beta)$ to $\OracleTable$ and feed $\beta$ to $\MalVerifier$ as the response from $\Proof$.
            \item Output $(r, \OracleTable)$.
        \end{enumerate}
    \end{construction}
\end{mdframed}

\begin{mdframed}[nobreak=true]
    \begin{construction}
        \label{construc:full-sim}
        A perfect zero-knowledge simulator for a PCP system $(\Prover,\Verifier)$ which is $\ProverLocality$-locally computable from a PZK-PCP $(\ZKProver, \ZKVerifier)$. Let $\ZKSimulator_{\MalVerifier}$ denote the simulator for $(\ZKProver, \ZKVerifier)$ for a malicious verifier $\MalVerifier$. 
        
        \noindent
        $\Simulator_{\MalVerifier}(x)$:
        \begin{enumerate}[nolistsep]
            \item Run $\ZKSimulator_{\Adversary_{\MalVerifier}}(x)$ (that is, simulate the view of  $\Adversary_{\MalVerifier}$ (\cref{construc:composed-adversary})) to obtain a query-answer set $\OracleTable_0$ for $\ZKProof$ and random coins $\Rand_0$.
            \item \label{step:run-composed-adv} Run $\Adversary_{\MalVerifier}(x, \Rand)$ (\cref{construc:composed-adversary}) using $\OracleTable_0$ to answer its queries to $\ZKProof$ to obtain $(r,\OracleTable)$.
            \item Output $(r, \OracleTable)$.
        \end{enumerate}
    \end{construction}
\end{mdframed}

\begin{proof}[Proof of \cref{lem:locally-computable-implies-zk}]
    We show that \cref{construc:full-sim} is a perfect zero-knowledge simulator for $(\Prover, \Verifier)$. We first analyse \cref{construc:composed-adversary}, and show that its output is identically distributed to $\View_{\MalVerifier, \Prover}$. Then we analyse \cref{construc:full-sim}, and show that provided $\Adversary_{\MalVerifier}$ does not break the query bound on $\ZKProof\gets\ZKProver$, we can simulate its view of $\ZKProof$. Combining these two facts yields the result. 
    
    Let $\MalVerifier$ be an arbitrary (possibly adaptive) malicious verifier of $(\Prover, \Verifier)$, and let $\Proof$ and $\ZKProof$ be random variables denoting the output of $\Prover$ and $\ZKProver$ respectively. Consider the hybrid simulator $\Adversary_{\MalVerifier}^{\ZKProof}$ defined in \cref{construc:composed-adversary}. As $\Prover$ is $\ProverLocality$-locally computable by $\LocalComp$, we have that $\Proof$ is identically distributed to $(\LocalComp^{\ZKProof}(\Query))_{\Query\in\Domain(\Proof)}$. Thus the output of \cref{construc:composed-adversary} is distributed identically to $\View_{\MalVerifier, \Prover}$.

    Now we analyse \cref{construc:full-sim}. As $(\ZKProver, \ZKVerifier)$ is a PZK-PCP, and $\Adversary_{\MalVerifier}$ is a verifier of $\ZKProof \gets \ZKProver$, there exists a simulator $\ZKSimulator_{\Adversary_{\MalVerifier}}$ for $\Adversary_{\MalVerifier}^{\ZKProof}$, whose output is identically distributed to $\View_{\Adversary_{\MalVerifier}, \ZKProver}(x,w)$ provided that $\Adversary_{\MalVerifier}$ makes at most $\QueryBound$ queries to $\ZKProof$. Note that for every query made by $\MalVerifier$, $\Adversary_{\MalVerifier}$ makes at most $\ProverLocality$ queries to $\ZKProof$. Thus, the output of $\ZKSimulator_{\Adversary_{\MalVerifier}}$ is identically distributed to $\View_{\Adversary_{\MalVerifier}, \ZKProver}(x,w)$ provided that $\MalVerifier$ makes at most $\QueryBound/\ProverLocality$ queries. Therefore, as $\Adversary_{\MalVerifier}(x)$ is receiving answers to its queries to $\ZKProof$ identically distributed to the real proof, by the first part, its output is identically distributed to $\View_{\MalVerifier, \Prover}$.
\end{proof}

\subsection{Alphabet reduction of PZK-PCPs}
\label{sec:alphabet-reduction}
In this section, we use the notion of $\ProverLocality$-locally computable proofs to show that standard alphabet reduction techniques preserve zero knowledge. In particular, we will show this for the alphabet reduction construction in \cite[Lemma 2.13]{BenSassonGHSV06}. For the sake of clarity, we include a self-contained presentation of their construction below (\cref{construc:alphabet-reduction-pcp}). 

\begin{lemma}[Zero-knowledge alphabet reduction]
\label{lem:alphabet-reduction}
    If a language $\Language$ has a PZK-PCP over the alphabet $\Bits^\AnswerLength$, query complexity $\NumQueries$, randomness complexity $\Rand$, query bound $\QueryBound$, decision complexity $\DecComplexity$ and robust-soundness error $\SoundErr$ with robustness parameter $\Robustness$, then $\Language$ has a boolean PZK-PCP with query complexity $O(\AnswerLength\cdot \NumQueries)$, randomness complexity $\Rand$, query bound $\QueryBound$, decision complexity $\DecComplexity+O(\AnswerLength\cdot\NumQueries)$ and robust-soundness error $\SoundErr$ with robustness parameter $\Omega(\Robustness)$. 
\end{lemma}

Following \cite{BenSassonGHSV06}, the construction relies on the existence of good error-correcting codes, which have constant relative distance and rate. To obtain the desired decision complexity, a code which is also computable by linear-sized circuits is used, such as \cite{Spielman96}.

\newcommand{\NonBooQueryAlg}{Q_a}
\newcommand{\NonBooDecAlg}{D_a}

\begin{mdframed}[nobreak=true]
    \begin{construction}
        \label{construc:alphabet-reduction-pcp}
        A boolean PZK-PCP for the language $\Language$, given a PZK-PCP $(\NonBooProver, \NonBooVerifier)$ over the alphabet $\Bits^\AnswerLength$ for $\Language$, where the verifier $\NonBooVerifier$ has decision algorithm $\NonBooDecAlg$ and query algorithm $\NonBooQueryAlg$. Let $\Ecc\colon \Bits^\AnswerLength\to\Bits^b$ where $b = O(\AnswerLength)$ be a systematic binary error-correcting code of constant relative minimum distance, computable by an explicit circuit of size $O(\AnswerLength)$.
        
        \noindent\textbf{Proof:}
        \begin{enumerate}[nolistsep]
            \item Run $\NonBooProver(x)$ to obtain a proof $\NonBooProof$.
            
            \item Define the proof oracle $\EccProof$, by $\EccProof(\QueryLoc) \eqdef \Ecc(\NonBooProof(\QueryLoc))$, for all $\QueryLoc\in \Domain(\NonBooProof)$.
            \item Output $(\NonBooProof, \EccProof$). 
        \end{enumerate}
        \vspace{0.5cm}
        \noindent\textbf{Verifier:}
        \begin{enumerate}[nolistsep] 
            \item Compute the query set $Q \eqdef \NonBooQueryAlg(x) \subseteq \Domain(\NonBooProof)$.
            \item Run $\NonBooDecAlg(x, \NonBooProof|_{Q})$ and reject if it rejects.
            \item For each $\gamma \in Q$, compute $\Ecc(\NonBooProof(\gamma))$, query $\EccProof(\gamma)$ and reject if $\Ecc(\NonBooProof(\gamma)) \neq \EccProof(\gamma)$.
            \item If all of the above checks pass, then accept.
        \end{enumerate}
    \end{construction}
\end{mdframed}

\begin{proof}[Proof of \cref{lem:alphabet-reduction}]
    As \cref{construc:alphabet-reduction-pcp} is entirely unchanged from \cite{BenSassonGHSV06}, a proof of completeness, soundness and all of the parameters (apart from the query bound) can be found in Lemma 2.13 of \cite{BenSassonGHSV06}. Here we argue zero knowledge.

    We show zero knowledge by demonstrating that $\Prover$, as defined in \cref{construc:alphabet-reduction-pcp}, is $1$-locally computable from the PZK-PCP $\NonBooProver$, which has query bound $\QueryBound$. Hence, by \cref{lem:locally-computable-implies-zk}, the PCP $\Prover$ inherits its zero-knowledge property from $\NonBooProver$. Towards this end, observe that $\Prover$ is $1$-locally computable by the function $\LocalComp\colon \Domain(\Proof) \to \Bits^b$, defined as
    \begin{equation*}
        f^{\NonBooProof}(\Oracle, \Query, i) = \begin{cases}
            \NonBooProof(\Query) ~&\text{if}~\Oracle = \NonBooProof
            \\
            \Ecc(\NonBooProof(\Query))_i ~&\text{if}~\Oracle = \EccProof
        \end{cases}
        .
    \end{equation*}
    Thus, by \cref{lem:locally-computable-implies-zk}, \cref{construc:alphabet-reduction-pcp} is perfect zero knowledge, with query bound $\QueryBound$.
\end{proof}

\subsection{Proof composition}
\label{sec:proof-comp}
We extend the proof composition theorem in \cite[Theorem 2.7]{BenSassonGHSV06} to account for zero knowledge. Recall that proof composition is performed with a robust outer PCP and an inner PCP of proximity. We show that if proof composition is carried out with a robust outer PCP that is zero knowledge, then the composed PCP will be zero knowledge. Crucially, the inner PCP of proximity is \emph{not} required to be zero knowledge. 

\begin{theorem}[Zero-knowledge proof composition]
\label{thm:zk-proof-comp}
    Suppose that for functions $\Rout, \Rin, \Dout, \Din, \Qin, \QueryBound\colon \N\to \N$ and $\Eout, \Ein,$ $\RobOut, \ProxIn \colon \N \to [0,1]$ the following hold:
    \begin{itemize}
        \item Language $\Language$ has a robust boolean PZK-PCP $(\ProverOut,\VerifierOut)$ with randomness complexity $\Rout$, query complexity $\Qout$, decision complexity $\Dout$, robust-soundness error $1 - \Eout$, robustness parameter $\RobOut$ and zero-knowledge query bound $\QueryBound$.
        \item $\CircVal$ has a boolean PCPP $(\ProverIn, \VerifierIn)$ where $\ProverIn$ is deterministic and efficient, with randomness complexity $\Rin$, query complexity $\Qin$, decision complexity $\Din$, proximity parameter $\ProxIn$, and soundness error $1-\Ein$.
        \item $\ProxIn(\Dout(n)) \leq \RobOut(n)$ for every $n$.
    \end{itemize}
    Then $\Language$ has a PZK-PCP, denoted $(\ProverComp,\VerifierComp)$ (cf. \cref{construc:comp-pcp}) with:
    \begin{itemize}
        \item randomness complexity $\Rout(n) + \Rin(\Dout(n))$;
        \item query complexity $\Qin(\Dout(n))$;
        \item decision complexity $\Din(\Dout(n))$;
        \item soundness error $1 - \Eout(n)\cdot \Ein(\Dout(n))$; and
        \item query bound $\QueryBound/\Qout$;
    \end{itemize}
\end{theorem}

\cref{construc:comp-pcp} is entirely unchanged from \cite{BenSassonGHSV06} (apart from the stipulation that the outer PCP be zero knowledge), and is reproduced  here for completeness.

\newcommand{\QueryAlgOut}{Q_{\text{out}}}
\newcommand{\QueryAlgIn}{Q_{\text{in}}}
\newcommand{\CktOut}{C_{\text{out}}}

\begin{mdframed}[nobreak=true]
    \begin{construction}[Composed PCP]
        \label{construc:comp-pcp}
        A composed PZK-PCP $(\ProverComp, \VerifierComp)$ for the language $\Language$, given a robust boolean outer PZK-PCP $(\ProverOut, \VerifierOut)$ for $\Language$ (with randomness complexity $\Rout$ and query complexity $\Qout$) and an inner boolean PCPP $(\ProverIn, \VerifierIn)$ for $\CircVal$ (with query complexity $\Qin$). Let $\DecisionCktOut$ and $\QueryAlgOut$  denote the decision and query algorithms for $\VerifierOut$ respectively. Both parties receive common input $x \in \Bits^n$.
        
        \noindent\textbf{Proof:}
        \begin{enumerate}[nolistsep]
            \item Run $\ProverOut(x)$ to obtain a proof $\ProofOut$.
            
            \item \label{step:comp-prover-2} For each choice of randomness $\Rand \in \Bits^{\Rout}$, compute the query set $\QuerySetOutR{\Rand} \eqdef \QueryAlgOut(x;\Rand)$.
            
            \item Compile $\DecisionCktOut$ into a circuit $\CktOut\colon \Bits^n \times \Bits^{\QueryOut} \times \Bits^{\Rout} \to \Bits$, and for each $\Rand \in \Bits^{\Rout}$, define the restricted circuit $C_{\text{out}, \Rand} \eqdef \CktOut(x,\cdot, \Rand)$.
            
            \item \label{step:comp-prover-3} For each $\Rand\in\Bits^{\Rout}$, run $\ProverIn(C_{\text{out},\Rand}, \ProofOut|_{\QuerySetOutR{\Rand}})$ to obtain a proof $\Proof_\Rand$.
            
            \item Output $\Proof \eqdef (\ProofOut, (\Proof_\Rand)_{\Rand\in\Bits^{\Rout}})$.
        \end{enumerate}
        \vspace{0.5cm}
        \noindent\textbf{Verifier:} 
        \begin{enumerate}[nolistsep]
            \item Compile $\DecisionCktOut$ into a circuit $\CktOut\colon \Bits^n \times \Bits^{\QueryOut} \times \Bits^{\Rout}\to\Bits$.

            \item Sample $\Rand\gets\Bits^{\Rout}$, and define the restricted circuit $C_{\text{out}, \Rand} \eqdef \CktOut(x,\cdot, \Rand)$.
            
            \item Run $\VerifierIn^{\ProofOut, \Proof_r}(C_{\text{out}, \Rand})$, treating $\ProofOut$ as the input oracle and $\Proof_\Rand$ as the proof oracle. Accept if and only if $\VerifierIn$ accepts.
            
        \end{enumerate}
    \end{construction}
\end{mdframed}

In order to show that proof composition preserves zero-knowledge, we will show that $\ProverComp$ (defined in \cref{construc:comp-pcp}) is locally computable from $\ProverOut$.

\begin{claim}
    \label{claim:comp-pcp-locality}
    Let $(\ProverComp, \VerifierComp)$ be the composed PCP system defined in \cref{construc:comp-pcp}, and let $(\ProverOut, \VerifierOut)$ be the perfect zero-knowledge outer PCP of $(\ProverComp, \VerifierComp)$. $\ProverComp$ is $\Qout$-locally computable from $\ProverOut$.
\end{claim}

\begin{proof}
    Clearly queries to the outer proof are $1$-locally computable. To locally compute queries to $\Proof_{\Rand}$ for any $\Rand \in \Bits^\Rout$, run \crefrange{step:comp-prover-2}{step:comp-prover-3} of the prover in \cref{construc:comp-pcp} to obtain a proof $\Proof_\Rand$, then output $\Proof_{\Rand}(\Query)$.
    Correctness is clear. In the worst case, we make $\max_{\Rand\in \Bits^{\Rout}}|\QuerySetOutR{\Rand}| = \Qout$ many queries to $\ProofOut$. Finally, the efficiency and determinism follow from the efficiency and determinism of $\ProverIn$ and $\VerifierOut$ (which is deterministic for any fixed choice of randomness). 
\end{proof}

\begin{proof}[Proof of \cref{thm:zk-proof-comp}]
    The randomness complexity, query complexity and decision complexity follow straightforwardly from the construction. A proof that \cref{construc:comp-pcp} satisfies completeness and the prescribed soundness error can be found in \cite{BGHSV05}. The construction is perfect zero-knowledge with query bound $\QueryBound/\Qout$ by \cref{lem:locally-computable-implies-zk} and \cref{claim:comp-pcp-locality}.
\end{proof}

We finish this section by instantiating the inner PCPP for $\CircVal$ to prove a corollary of \cref{thm:zk-proof-comp}, in which the parameters of the composed ZKPCP only depend on the parameters of the outer ZKPCP. We shall need Theorem 3.3 in \cite{BGHSV05}, which we restate below, setting the parameter $t(n) = 2/\varepsilon$, while noting that the honest prover strategy is efficient.
\begin{theorem}[\cite{BGHSV05}]
\label{thm:ckt-val-pcpp}
    For any $\varepsilon > 0$, $\CircVal$ has a PCPP $(\Prover, \Verifier)$ where $\Prover$ is deterministic and efficient, with randomness complexity $\log n + O(\log^\varepsilon(n))$, query complexity $O(1/\varepsilon)$, proximity parameter $\Theta(\varepsilon)$ and soundness error $1/2$.
\end{theorem}

\begin{corollary}
\label{cor:zk-proof-comp}
	$\RobustPZKPCP_{\QueryBound, \Bits}[r,q] \subseteq \PZKPCP_{\QueryBound/q, \Bits}[r+\log n, 1]$.
\end{corollary}
\begin{proof}
    Let $\Language \in \RobustPZKPCP_{\QueryBound, \Bits}[r,q]$. That is, there is a PZK-PCP $(\Prover, \Verifier)$ for $\Language$ with randomness complexity $r(n)$, query complexity $q(n)$, query bound $\QueryBound(n)$ and constant robustness $\Robustness$. \cref{thm:ckt-val-pcpp} guarantees the existence of a PCPP for $\CircVal$ with proximity parameter $k\varepsilon$, for some universal constant $k$ and any $\varepsilon>0$. Set $\varepsilon \eqdef \Robustness/k$.
	Then by \cref{thm:zk-proof-comp} (note that our choice of $\varepsilon$ ensures that the third bullet of the hypothesis is met), $\Language$ has a PZK-PCP with randomness complexity $r(n) + \log(n) + O(\log^\varepsilon(n))$ and query complexity $O(1/\varepsilon)$, as required.
\end{proof}

\section{Robust PCPPs for polynomial summation}
\label{sec:rob-sumcheck}
In this section, we construct a PCP of proximity for the $\SCLang$ pair language, defined below, which achieves constant \emph{robust} soundness.

\begin{definition}
\label{def:sc-lang}
	We define the language $\SCLang$ to be the set of all pairs $((\Field,1^\NumVars,1^\Degree,\Subcube,\SumcheckTotal), F)$
    where $\Field$ is a finite field, $\NumVars,\Degree \in \N$, $\Subcube \subseteq \Field$, $\SumcheckTotal \in \Field$, and
	$F \in \ReedMuller[\Field,\NumVars,\Degree]$ with $\sum_{\vec b \in \Subcube^\NumVars} \SCPoly(\vec b) = \SumcheckTotal$.
\end{definition}

Note that $\Degree$ and $\NumVars$ are specified in unary, so that the verifier and simulator for our PCPP can run in time $\Poly(\log |\Field|, \NumVars, \Degree, |\Subcube|)$. Below, for readability, we will write explicit inputs of $\SCLang$ as $(\Field, \NumVars, \Degree, \Subcube, \SumcheckTotal)$; this should be understood as $(\Field, 1^\NumVars, 1^\Degree, \Subcube, \SumcheckTotal)$.

\begin{lemma}
\label{lem:robust-sumcheck-pcp}
    Let $\Field$ be a finite field, and let $\NumVars, \Degree \in \N$, $\Subcube \subseteq \Field$ and $\delta > 0$ be such that $\Proximity > \frac{\NumVars \Degree}{|\Field|}$ and $\Degree \geq |\Subcube|+1$. Then \cref{construc:robust-sumcheck} is a robust PCP of proximity for $\SCLang$, over the alphabet $\Field^{\NumVars-1}$, with proximity parameter $\Proximity$ and robustness parameter $\Robustness = \Omega(\Proximity)$. The verifier makes $O(|\Field|)$-many queries to $\SumcheckInput$ and $|\Field|$-many queries to $\Proof$. The proof length is $|\Field|^{\NumVars-1}$.
\end{lemma}

While the robust PCP of \cite{BenSassonGHSV06} uses (a variant of) sumcheck, to the best of our knowledge \cref{lem:robust-sumcheck-pcp} has not previously been shown explicitly.
Our proof here is substantially simpler than that of \cite{BenSassonGHSV06}, in part because the structure of our PCP permits a simpler ``bundling'' of symbols. We will build on this construction when constructing robust PZK-PCPPs for $\SCLang$, in \cref{sec:robust-pzk-pcp}.

\begin{mdframed}[nobreak=true]
    \begin{construction}
    \label{construc:robust-sumcheck}
        A robust PCP of proximity $\RobSC$ for $\SCLang$ with proximity parameter $\Proximity >0$ over the alphabet $\Field^{\NumVars-1}$. Both parties receive the explicit input $(\Field, \NumVars, \Degree, \Subcube, \SumcheckTotal, \Proximity)$, and oracle access to the evaluation table of the implicit input $\SumcheckInput\colon \Field^\NumVars \to \Field$.

        \noindent\textbf{Proof:}
        \begin{enumerate}[nolistsep]
            \item For each $i \in [\NumVars-1]$, define the $i$-th sumcheck polynomial to be the $i$-variate polynomial $\ProverFunc_i(\vec{X}) \in \Polys{\Field}{\Degree}{i}$ given by
            \begin{equation*}
                \ProverFunc_i(x_1,\dots,x_i) \eqdef \sum_{\vec{b}\in \Subcube^{\NumVars-i}}\SumcheckInput(x_1,\dots,x_i, \vec{b}).
            \end{equation*}
            
            \item Define $\Proof\colon\Field^{\NumVars-1}\to \Field^{\NumVars-1}$ by 
            \begin{equation*}
                \Proof(\VerRand_1,\dots,\VerRand_{\NumVars-2}, \Query) \eqdef (\ProverFunc_1(\Query), \ProverFunc_2(\VerRand_1,\Query), \dots, \ProverFunc_{\NumVars-1}(\VerRand_1,\dots,\VerRand_{\NumVars-2}, \Query)),
            \end{equation*}
            for every $(\VerRand_1,\dots,\VerRand_{\NumVars-2}, \Query)\in \Field^{\NumVars-1}$.

            \item Output $\Proof$.

        \end{enumerate}

        \noindent\textbf{Verifier:}
        \begin{enumerate}[nolistsep]
            \item Sample $\VerRandVec\eqdef(\VerRand_1,\dots,\VerRand_{\NumVars-1})\gets\Field^{\NumVars-1}$ uniformly at random.

            \item For all $\Query\in\Field$, query $\Proof(\VerRand_1,\dots,\VerRand_{\NumVars-2},\Query) \eqdef (\ProverFunc_1(\Query), \ProverFunc_2(\VerRand_1, \Query),\dots, \ProverFunc_{\NumVars-1}(\VerRand_1,\dots,\VerRand_{\NumVars-2}, \Query))$.

            \item For all $\Query\in\Field$, query $\SumcheckInput(\VerRand_1,\dots,\VerRand_{\NumVars-1},\Query)$.

            \item \label{step:check-proof-ld} For each $i \in [\NumVars-1]$, reject if $\ProverFunc_i(\VerRand_1,\dots,\VerRand_{i-1}, X)$ is not a polynomial of degree at most $\Degree$.

            \item \label{step:sumcheck-test1} Check that 
            \begin{equation*}
                \sum_{b \in \Subcube}g_1(b) = \SumcheckTotal.
            \end{equation*}

            \item \label{step:sumcheck-test2} For each $i \in [\NumVars-2]$, check that 
            \begin{equation*}
                \sum_{b\in \Subcube} \ProverFunc_{i+1}(\VerRand_1,\dots,\VerRand_i, b) = \ProverFunc_i (\VerRand_1,\dots,\VerRand_i).
            \end{equation*}

            \item \label{step:sumcheck-test3} Check that 
            \begin{equation*}
                \sum_{b \in \Subcube} \SumcheckInput(\VerRandVec, b) = \ProverFunc_{\NumVars-1}(\VerRandVec).
            \end{equation*}

            \item \label{step:sumcheck-ldt}Perform a robust low total degree test (\cref{thm:bundled-ldt}) on $\SumcheckInput$, with proximity parameter $\RMDist = \min(\Proximity, 1/5)$, and reject if the test fails.

            \item If none of the above checks fail, then accept.
        \end{enumerate}
    \end{construction}
\end{mdframed}

We now prove that \cref{construc:robust-sumcheck} is robust, but first outline our strategy for doing so. Conceptually speaking, \cref{construc:robust-sumcheck} consists of two tests: the sumcheck tests in \crefrange{step:check-proof-ld}{step:sumcheck-test3} (which includes a check that the proof itself has low-degree structure), and the low-degree test on the input oracle in \cref{step:sumcheck-ldt}. The low-degree test is indeed robust (by \cref{thm:bundled-ldt}), so the key technical component in proving the robustness of \cref{construc:robust-sumcheck} is \cref{lem:robust-sumcheck-tech}, which asserts that the sumcheck tests (\crefrange{step:check-proof-ld}{step:sumcheck-test3}) reject robustly, under the assumption the input oracle is close to a low-degree polynomial. Since the verifier makes roughly the same number of queries when performing both of these tests, this implies the overall construction is robust.

\cref{lem:robust-sumcheck-tech} holds due to the assumed low-degree structure of $F$ and our specific choice of bundling. If the view of the sumcheck verifier is rejecting, then, since both the proof $\Proof$ and the input oracle $F$ are close to low degree, at least one of $F, g_1,\dots, g_{\NumVars-1}$ needs to be modified to equal a different polynomial. By the Schwartz-Zippel lemma, this requires a substantial fraction of modifications. Our bundling strategy ensures that these modifications form a substantial fraction of the view of the verifier. We now proceed with a formal proof.

\begin{lemma}
\label{lem:robust-sumcheck-tech}
    Let $\BadF\colon \Field^\NumVars \to \Field$ be $\SumDist$-far from $\SCLang[\Field, \NumVars, \Degree, \Subcube, \SumcheckTotal]$, but $\RMDist$-close to $\ReedMuller[\Field, \NumVars, \Degree]$ for some $\frac{\NumVars \Degree}{|\Field|} < \RMDist \leq \SumDist$ and $\Degree \geq |\Subcube|+1$. Then for all proofs $\MalProof$,
    \begin{equation*}
        \mathop{\mathbb{E}}_{\VerRandVec\gets \Field^{\NumVars-1}}\left[\HammingDist\left((\MalProof(\VerRand_1,\dots,\VerRand_{\NumVars-2}, \Query)_{\Query\in\Field}, \BadF(\VerRand_1,\dots,\VerRand_{\NumVars-1},\Query)_{\Query\in \Field}), ~\VAcc(\Verifier)\right)\right] \geq \frac{1}{2}\min\left(\RMDist, 1-4\RMDist\right).
    \end{equation*}
\end{lemma}

\begin{proof}
    Let $\hat F \in \Polys{\Field}{\Degree}{\NumVars}$ be the some polynomial of total degree at most $\Degree$ which is $\RMDist$-close to $\BadF$, as guaranteed by the hypothesis. Note that by the condition that $\BadF$ is $\SumDist$-far from $\SCLang[\Field, \NumVars, \Degree, \Subcube, \SumcheckTotal]$ and $\RMDist \leq \SumDist$, it must be that $\sum_{\vec{b}\in \Subcube^{\NumVars}} \hat F(\vec{b}) \neq \SumcheckTotal$. Before proceeding, we remind the reader of some notation: for any $\VerRandVec = (\VerRand_1,\dots,\VerRand_{\NumVars-1}) \in \Field^{\NumVars-1}$, and any $i \in [\NumVars-1]$, we denote $\VerRandVec_i \eqdef (\VerRand_1,\dots,\VerRand_i)$. Also, for a polynomial $f \in \Polys{\Field}{\Degree}{i}$, we denote the restriction of $f$ to the axis-parallel line defined by $\VerRandVec_{i-1}$ by $f|_{\VerRandVec_{i-1}}(X) \eqdef f(\VerRandVec_{i-1}, X)$.

    First, fix some $\VerRandVec = (\VerRand_1, \dots, \VerRand_{\NumVars-1}) \in \Field^{\NumVars-1}$, and  suppose that for some $i \in [\NumVars-1]$, $\ProverFunc_i(\VerRand_1,\dots,\VerRand_{i-1}, X)$ is $\RMDist$-far from every univariate polynomial of degree at most $\Degree$. Then since the verifier inspects each $\ProverFunc_i(\VerRand_1,\dots,\VerRand_{i-1}, X)$ at every point in $\Field$ and rejects if $\ProverFunc_i(\VerRand_1,\dots,\VerRand_{i-1}, X)$ is not a polynomial of degree at most $\Degree$ (\cref{step:check-proof-ld}), at least a $\RMDist$-fraction of symbols of $\MalProof(\VerRand_1,\dots,\VerRand_{\NumVars-2}, \Query)_{\Query\in\Field}$ would need to be modified in order for the verifier to accept. Thus $(\MalProof(\VerRand_1,\dots,\VerRand_{\NumVars-2}, \Query)_{\Query\in\Field}, \BadF(\VerRand_1,\dots,\VerRand_{\NumVars-1},\Query)_{\Query\in \Field})$ has distance at least $\RMDist/2$ to any accepting view.
    
    Hence we can proceed under the assumption that for all $i \in [\NumVars-1]$, $\ProverFunc_i(\VerRand_1,\dots,\VerRand_{i-1}, X)$ is $\RMDist$-close to some univariate polynomial $\ProverPoly_i(X)\in \Field^{\leq \Degree}[X]$. Let $\Event$ be the event that either
    \begin{equation*}
         \sum_{b \in \Subcube}\ProverPoly_1(b) \neq \SumcheckTotal,
    \end{equation*}
    or
    \begin{equation*}
    \label{eqn:ver-rej-proof}
        \sum_{b\in \Subcube} \ProverPoly_{i+1}(b) \neq \ProverPoly_i (\VerRand_i)
    \end{equation*}
    for some $i \in [\NumVars-2]$, or 
    \begin{equation*}
    \label{eqn:ver-rej-inp}
        \sum_{b \in \Subcube} \hat\SumcheckInput(\VerRandVec, b) \neq \ProverPoly_{\NumVars-1}(c_{m-1}).
    \end{equation*}
    Note that since $\sum_{\vec{b}\in \Subcube^{\NumVars}} \hat F(\vec{b}) \neq \SumcheckTotal$, by the soundness of sumcheck, $\Pr_{\VerRandVec}[\neg \Event] \leq \NumVars \Degree/|\Field|$.

    Conditioned on $\Event$, in order for the verifier to accept, one of $\ProverFunc_1(X), \ProverFunc_2|_{\VerRandVec_{1}}(X), \dots, \ProverFunc_{\NumVars-1}|_{\VerRandVec_{\NumVars-2}}(X), \BadF|_{\VerRandVec}(X)$ must be modified to equal a degree $\Degree$ polynomial other than $\ProverPoly_1(X), \ProverPoly_2|_{\VerRandVec_{1}}(X), \dots, \ProverPoly_{\NumVars-1}|_{\VerRandVec_{\NumVars-2}}(X), \hat{F}|_{\VerRandVec}(X)$ respectively. This would require modifying $\ProverFunc_i|_{\VerRandVec_{i-1}}(X)$, for some $i \in [\NumVars-1]$, in at least a $1 - \Degree/|\Field| - \HammingDist(\ProverFunc_i|_{\VerRandVec_{i-1}} ,\ProverPoly_i|_{\VerRandVec_{i-1}}) \geq 1 - \Degree/|\Field| -\RMDist$ fraction of points, or modifying $\BadF_{\VerRandVec}(X)$ in at least a $1 - \Degree/|\Field| - \HammingDist(\BadF|_{\VerRandVec}, \hat{F}|_{\VerRandVec})$ fraction of points. Thus, the verifier's view of the proof would require modification in at least a $\frac{1}{2}(1 - \frac{\Degree}{|\Field|} - \max(\RMDist, \HammingDist(\BadF|_{\VerRandVec}, \hat{F}|_{\VerRandVec})))$ fraction of points. 

    Thus the expected distance of $(\MalProof(\VerRand_1,\dots,\VerRand_{\NumVars-2}, \Query)_{\Query\in\Field}, \BadF(\VerRand_1,\dots,\VerRand_{\NumVars-1},\Query)_{\Query\in \Field})$ to an accepting view is at least 
    \begin{align*}
        \frac{1}{2} \cdot\mathop{\mathbb{E}}_{\VerRandVec\gets \Field^{\NumVars-1}}\left[1 - \frac{\Degree}{|\Field|} - \max(\RMDist, \HammingDist(\BadF|_{\VerRandVec}, \hat{F}|_{\VerRandVec}))\right] - \Pr_{\vec{c}}[\neg \Event] &=
        \frac{1}{2} \left(1 - \frac{\Degree}{|\Field|} - \RMDist\right) - \Pr_{\vec{c}}[\neg \Event]
        \\
        &\geq \frac{1}{2}\left(1 - \frac{\Degree}{|\Field|} - \RMDist\right) - \frac{\NumVars\Degree}{|\Field|}
        \\
        &\geq \frac{1}{2}(1 - 4\RMDist),
    \end{align*}
    where we have used that 
    $\mathbb{E}[\HammingDist(\BadF|_{\VerRandVec}, \hat{F}|_{\VerRandVec}))] \leq \RMDist$ and that
    $\frac{\NumVars \Degree}{|\Field|} \leq \RMDist$ by hypothesis.
\end{proof}

Now we employ \cref{lem:robust-sumcheck-tech} to prove that \cref{construc:robust-sumcheck} is robust, following the approach we outlined above.

\begin{proof}[Proof of \cref{lem:robust-sumcheck-pcp}]
    Completeness is straightforward; we prove robust soundness and efficiency.
    
    \parhead{Robust soundness} Suppose that $\BadF$ is $\Proximity$-far from  $\SCLang[\Field, \NumVars, \Degree, \Subcube, \SumcheckTotal]$. First, if $\SumcheckInput$ is $\RMDist$-far from $\ReedMuller[\Field, \NumVars,\Degree]$, then by \cref{thm:bundled-ldt}, the expected distance of the view of the low degree test (\cref{step:sumcheck-ldt}) from any accepting view is $\Omega(\Proximity)$ (since $\RMDist \leq \Proximity$). As the low-degree test makes $O(|\Field|)$ queries to $\SumcheckInput$, and (apart from the low-degree test) the verifier in \cref{construc:robust-sumcheck} makes $|\Field|$ queries to each of $\SumcheckInput$ and $\Proof$, the queries made by the low-degree test form a constant fraction of the queries made overall by the verifier. Hence the expected distance of the view of the verifier from any accepting view is also $\Omega(\Proximity)$.

    Second, if $\SumcheckInput$ is $\Proximity$-far from $\SCLang[\Field, \NumVars, \Degree, \Subcube, \SumcheckTotal]$ but $\RMDist$-close to $\ReedMuller[\Field, \NumVars, \Degree]$, then by \cref{lem:robust-sumcheck-tech}, the expected distance of the view of the verifier (when ignoring the low-degree test) to any accepting view is at least $\frac{1}{2}\min(\RMDist, (1-4\RMDist)) \geq \min(\delta/2,1/10) = \Omega(\Proximity)$ (since $\RMDist \eqdef \min(\Proximity, 1/5)$).
    Once again, the queries made by the verifier apart from the low-degree test form a constant fraction of the overall number of queries made by the verifier, so the expected distance of the view of the verifier from any accepting view is $\Omega(\Proximity)$. In other words, \cref{construc:robust-sumcheck} has expected robustness $\Omega(\Proximity)$. 
    
    By \cref{prop:exp-rob-to-rob}, for an appropriate choice of constant $\varepsilon$, \cref{construc:robust-sumcheck} has robustness $\Omega(\Proximity)$.
    
    \parhead{Efficiency} The proof length is clear from construction. The query complexity follows from the fact that the low-degree test (\cref{step:sumcheck-ldt}) makes $O(|\Field|)$ queries to $F$, and the rest of \cref{construc:robust-sumcheck} makes $|\Field|$ queries to each of $F$ and $\Proof$.
\end{proof}

\section{Robust PZK-PCPPs for polynomial summation}
\label{sec:robust-pzk-pcp}
In this section, we construct PZK-PCPs of proximity for $\SCLang$ (\cref{def:sc-lang}) with constant robust soundness. The construction makes black-box use of the robust sumcheck PCPP we obtain in \cref{lem:robust-sumcheck-pcp} and builds on the PZK-PCP for $\SharpP$ constructed in \cite{GurOS2024}. 

\begin{lemma}
    \label{lem:rob-zk-pcp-sumcheck}
    Let $\Proximity > 0$, $\Field$ be a finite field, $\Subcube\subseteq \Field$, $\SumcheckTotal \in \Field$, and $\NumVars, \Degree \in \N$ be such that $\frac{\NumVars \Degree}{|\Field|} < \Proximity$ and $\Degree \geq |\Subcube| + 1$. \cref{construc:robust-sumcheck-pzkpcp} is a robust (black-box) perfect zero-knowledge PCP of proximity for $\SCLang[\Field,\NumVars,\Degree,\Subcube,\gamma]$ over the alphabet $\Field^{\NumVars+1}$ with proximity parameter $\Proximity$ and robustness parameter $\Robustness = \Omega(\Proximity)$. The verifier complexity is $\Poly(|\Field|, \NumVars, \Degree)$, the 
    verifier makes $O(|\Field|)$ queries to both $F$ and $\Proof$, and the proof length is $O(|\Field|^{\NumVars})$.
\end{lemma}

We note that (as in \cite{GurOS2024}) the zero-knowledge property for this PCPP holds against all polynomial-time malicious verifiers regardless of their query complexity (in fact, under an appropriate definition, ZK holds against \emph{all} verifiers regardless of their complexity).

\cref{construc:robust-sumcheck-pzkpcp} closely follows \cite{GurOS2024}, with only a few modifications to achieve constant robustness. We provide a brief high-level description of the \cite{GurOS2024} construction, and then present our construction in detail. We refer the reader to \cite{GurOS2024} for further exposition.

The \cite{GurOS2024} ZKPCP for $\SCLang$ is a (standard) sumcheck PCPP for a ``masked'' version of the polynomial $\SCPoly$; specifically, a sumcheck PCPP for the claim
\begin{equation}
	\label{eq:gos-claim}
	\sum_{\vec a \in \Subcube^\NumVars} \SCPoly(\vec a) + \underbrace{\RandPoly(\vec a) - \RandPoly(\revvec a)}_{\text{antisymmetric}} + \underbrace{\sum_{i=1}^{\NumVars} \ZeroP_\Subcube(a_i) \RandLDPoly_i(\vec a)}_{\text{nullstellensatz}} = \gamma
\end{equation}
where
\begin{itemize}[itemsep=2pt]
	\item $\RandPoly$ is a uniformly random polynomial in $\Polys{\Field}{\vec \Degree}{\NumVars}$;
	\item $\ZeroP_\Subcube \eqdef \prod_{a \in \Subcube} (X - a)$ is the nonzero univariate polynomial of minimal degree that is zero on $\Subcube$; and
	\item $\RandLDPoly_i$, for each $i \in [\NumVars]$, is a uniformly random polynomial in $\Polys{\Field}{\vec \Degree_i}{\NumVars}$, where $\vec \Degree_i$ taking value $\Degree - |\Subcube|$ in the $i$-th coordinate and $\Degree$ in all other coordinates.
\end{itemize}
The polynomials $Q, T_1,\ldots,T_\NumVars$ are provided in the proof. It is not difficult to see that, for any functions $Q, T_1,\ldots,T_\NumVars$, \eqref{eq:gos-claim} is true if and only if $\sum_{\vec a \in \Subcube^\NumVars} F(\vec a) = \gamma$.

The mask polynomials make the PCPP zero knowledge. The purpose of the ``antisymmetric'' term $Q(\vec X) - Q(\revvec X)$ is to hide the values of intermediate sums in the sumcheck proof, which are $\SharpP$-hard to compute, while summing to zero over $\Subcube^\NumVars$ to maintain soundness.
The ``nullstellensatz'' term computes, by the combinatorial nullstellensatz \cite{Alon99}, a uniformly random polynomial $N$ such that $N(\vec{h}) = 0$ for all $\vec{h} \in \Subcube$. By adding this term, we retain the antisymmetric structure of the mask on $\Subcube^\NumVars$ (necessary for soundness) while maximising the randomness outside of $\Subcube^\NumVars$ (necessary for simulation).

We are now ready to present our construction.

\begin{mdframed}[nobreak=true]
    \begin{construction}
    \label{construc:robust-sumcheck-pzkpcp}
    A robust PZK-PCP of proximity for $\SCLang[\Field,\NumVars,\Degree,\Subcube,\gamma]$ with proximity parameter $\Proximity>0$ over the alphabet $\Field^{\NumVars+1}$. Both parties receive common input $(\Field, \NumVars, \Degree, \Subcube, \SumcheckTotal, \Proximity)$, and oracle access to the evaluation table of $\SumcheckInput\colon \Field^\NumVars \to \Field$.
    
         \noindent\textbf{Proof:} 
    \begin{enumerate}[nolistsep]
        \item  Sample a polynomial $\RandPoly \gets \Polys{\Field}{\vec{\Degree}}{\NumVars}$ uniformly at random, where $\vec \Degree = (\Degree, \ldots, \Degree)$. 

        \item For each $i \in [\NumVars]$, sample $\RandLDPoly_i \gets \Polys{\Field}{\vec{\Degree}_i}{\NumVars}$ uniformly at random, where $\vec{\Degree}_i = (\Degree, \dots, \Degree-|\Subcube|, \dots, \Degree)$ is the vector which takes the value $\Degree$ in every coordinate except for the $i$-th, which takes the value $\Degree-|\Subcube|$. 

        \item Define $\Proof_P(\vec{x}) = (\RandPoly(\vec{x}), \RandLDPoly_1(\vec{x}), \dots, \RandLDPoly_\NumVars(\vec{x}))$ for all $\vec{x} \in \Field^\NumVars$, written over the alphabet $\Field^{\NumVars+1}$.

        \item Define $\MaskPoly(\vec{X}) \eqdef \RandPoly(\vec{X}) - \RandPoly_\rev(\vec{X}) + \sum_{i=1}^{\NumVars} \ZeroP_\Subcube(X_i) \RandLDPoly_i(\vec{X})$, where $\ZeroP_\Subcube \eqdef \prod_{a \in \Subcube} (X - a)$, and $\RandPoly_\rev(\vec{X}) \eqdef \RandPoly(\revvec{X})$.

        \item Compute $\Proof_\Sigma \eqdef \RobSC[(\Field, \NumVars, \Degree, \Subcube, \SumcheckTotal, \Proximity), F+R]$, where $\RobSC[x, f]$ denotes the robust sumcheck PCP with explicit input $x$ and implicit input $f$ (\cref{construc:robust-sumcheck}). Write this proof over the alphabet $\Field^{\NumVars+1}$ (rather than $\Field^{\NumVars-1}$) by padding each symbol with two zeros at the end.

        \item Output $\Proof \eqdef (\Proof_\Sigma,\Proof_P)$, written over the alphabet $\Field^{\NumVars+1}$.
    \end{enumerate}
    \vspace{0.5cm}
    \noindent\textbf{Verifier:} 
    \begin{enumerate}[nolistsep]
        \item \label{step:zk-rob-sumcheck-step}Emulate the robust sumcheck verifier $\Verifier_\Sigma$ (\cref{construc:robust-sumcheck}) on input $\SumcheckInput+\MaskPoly$  and proof $\Proof_\Sigma$. To query the input $\SumcheckInput+\MaskPoly$ at some point $\vec{\alpha}\in \Field^\NumVars$, query $F(\vec{\alpha})$, $\Proof_{P}(\vec\alpha)\eqdef (\RandPoly(\vec{\alpha}), \RandLDPoly_1(\vec{\alpha}), \dots, \RandLDPoly_\NumVars(\vec{\alpha}))$ and $\Proof_P(\revvec \alpha)\eqdef (\RandPoly(\revvec{\alpha}), \RandLDPoly_1(\revvec{\alpha}), \dots, \RandLDPoly_\NumVars(\revvec{\alpha}))$, then compute $(F+R)(\vec\alpha) \eqdef F(\vec{\alpha}) + \RandPoly(\vec{\alpha}) - \RandPoly(\revvec{\alpha}) + \sum_{i=1}^{\NumVars} \ZeroP_\Subcube(\alpha_i) \RandLDPoly_i(\vec{\alpha})$. Reject if $\Verifier_\Sigma$ rejects.
        
        \item Perform a robust low-degree test (\cref{thm:bundled-ldt}) on $\SumcheckInput$, with proximity parameter $\RMDist \eqdef \min(\Proximity, 1/5)$, and reject if the test fails.

        \item \label{step:bundled-ldt} Perform a robust vector-valued low-degree test on $\Proof_P$ (\cref{thm:bundled-ldt}), with proximity parameter\\ $\LDTProx \eqdef \frac{1}{8}\RMDist$ and degree parameter $d_P = \NumVars\Degree$, and reject if the test fails. 
        \item Accept if and only if none of the above tests rejected.
    \end{enumerate}
    \end{construction}
\end{mdframed}

\newcommand{\MalSCProof}{\MalProof_\Sigma}
\newcommand{\MalPolyProof}{\MalProof_P}

To prove the robustness of our construction, we show a technical lemma (see \cref{lem:pzk-sumcheck-tech}) which asserts that the PCPP is robust under the assumption that the input oracle is close to some Reed-Muller codeword that does not sum to $\gamma$ over $\Subcube^\NumVars$ (and indeed is far from all Reed-Muller codewords that do). This will imply that the whole construction is robust, since the low-degree test is robust.

We also need to argue that the zero-knowledge property holds. We accomplish this using the notion of locally computable proofs developed in \cref{sec:loc-proofs}. Essentially, we show that the PCP in \cref{construc:robust-sumcheck-pzkpcp} is locally computable from the perfect zero-knowledge PCPP for sumcheck defined in \cite{GurOS2024}, and so it inherits zero knowledge from the GOS construction by \cref{lem:locally-computable-implies-zk}.

Throughout the remainder of this section we will use the following notation. We denote the non-ZK robust sumcheck PCP verifier (\cref{construc:robust-sumcheck}) by $\Verifier_{\Sigma}$. We denote the verifier that performs \cref{step:zk-rob-sumcheck-step} of \cref{construc:robust-sumcheck-pzkpcp} by $\Verifier_1$. That is, $\Verifier_1$ is given $F\colon \Field^{\NumVars} \to \Field$ as input, and runs $\Verifier_\Sigma$ on input oracle $F+R\colon\Field^\NumVars\to\Field$ (computing queries to $F+R$ via $F$ and the proof $\Proof_P$) and proof oracle $\Proof_\Sigma$. $\Verifier_1$ then outputs the output of $\Verifier_{\Sigma}$. For any $\VerRandVec \in \Field^{\NumVars-1}$, we denote the view of $\Verifier_{\Sigma}$, on randomness $\VerRandVec \in \Field^{\NumVars-1}$ and input $G\colon \Field^{\NumVars} \to\Field$, by $\SCView(\VerRandVec)\eqdef (\Proof_{\Sigma}(\VerRandVec_{\NumVars-2}, \alpha)_{\alpha \in \Field}, G(\VerRandVec, \alpha)_{\alpha \in \Field})$ and we denote the view of $\Verifier_1$, on randomness $\VerRandVec\in\Field^{\NumVars-1}$ and input $\SumcheckInput\colon \Field^{\NumVars}\to\Field$, by $\HonView(\VerRandVec) \eqdef (\Proof_{\Sigma}(\VerRandVec_{\NumVars-2}, \Query)_{\Query\in\Field},  \Proof_P(\VerRandVec, \Query)_{\Query\in \Field}, \Proof_P(\Query, \revvec{\VerRand})_{\Query\in \Field}, \SumcheckInput(\VerRandVec,\Query)_{\Query\in \Field})$.

\begin{lemma}
    \label{lem:pzk-sumcheck-tech}
    Let $\BadF \colon \Field^\NumVars \to \Field$ be $\SumDist$-far from  $\SCLang[\Field, \NumVars, \Degree, \Subcube, \SumcheckTotal]$, but $\RMDist$-close to $\ReedMuller[\Field, \NumVars, \Degree]$ for $\SumDist > \RMDist \geq \frac{\NumVars \Degree}{|\Field|}$, $\RMDist < 1/5$ and $\Degree \geq |\Subcube| + 1$. Then for all proofs $(\MalSCProof, \MalPolyProof)$,
    \begin{equation*}
        \mathop{\mathbb{E}}_{\VerRandVec\gets \Field^{\NumVars-1}}\left[\HammingDist\left((\MalSCProof(\VerRandVec_{\NumVars-2}, \Query)_{\Query\in\Field},  \MalPolyProof(\VerRandVec, \Query)_{\Query\in \Field}, \MalPolyProof(\Query, \revvec{\VerRand})_{\Query\in \Field},
        \BadF(\VerRandVec,\Query)_{\Query\in \Field}),~ \VAcc(\Verifier)\right)\right] =  \Omega(\RMDist).
    \end{equation*}
\end{lemma}
Here we view $(\MalSCProof(\VerRandVec_{\NumVars-2}, \Query)_{\Query\in\Field},  \MalPolyProof(\VerRandVec, \Query)_{\Query\in \Field}, \MalPolyProof(\Query, \revvec{\VerRand})_{\Query\in \Field},\BadF(\VerRandVec,\Query)_{\Query\in \Field}) \in (\Field^{\NumVars+1})^{4\Field}$ as a string over the alphabet $\Field^{\NumVars+1}$, of length $4|\Field|$, via some natural embedding of $\Field$ into $\Field^{\NumVars+1}$. The metric $\HammingDist$ refers to Hamming distance over $\Field^{\NumVars}$.

\begin{proof}
    First, suppose that $\MalPolyProof$ is $\LDTProx$-far (where $\LDTProx \eqdef \frac{1}{8}\RMDist$, see \cref{step:bundled-ldt}), with respect to the alphabet $\Field^{\NumVars+1}$, from $\ReedMuller^{\NumVars+1}[\Field,\NumVars,\NumVars \cdot \Degree]$. In this case, then by the robustness of the vector-valued low-degree test (\cref{thm:bundled-ldt}), the view of the verifier in \cref{step:bundled-ldt} has expected distance $\Omega(\LDTProx) = \Omega(\RMDist)$ to any accepting view. Since this test forms a constant fraction of the overall view of the verifier, the expected distance of the view of the verifier from any accepting view would be $\Omega(\RMDist)$.

    Hence we can assume that $\MalPolyProof$ is $\LDTProx$-close, with respect to the alphabet $\Field^{\NumVars+1}$, to some vector of $\NumVars$-variate polynomials each of degree at most $\NumVars\Degree$, which we denote by $(\hat\RandPoly(\vec{X}), \hat\RandLDPoly_1(\vec{X}), \dots,\hat\RandLDPoly_\NumVars(\vec{X}))$. We denote the ``corrected low-degree'' proof by
    \begin{equation*}
        \ProofLD(\vec{x}) \eqdef (\hat\RandPoly(\vec{x}), \hat\RandLDPoly_1(\vec{x}), \dots, \hat\RandLDPoly_\NumVars(\vec{x}))
    \end{equation*}
    and note that by definition of $\ProofLD$ we have $\HammingDist(\ProofLD, \Proof_P) \leq \LDTProx$, where this distance is with respect to $\Field^{\NumVars+1}$. Lastly, we denote $\Verifier_1$'s view of the ``corrected'' proof $\ProofLD$ by 
    \begin{equation*}
        \LDView(\VerRandVec) \eqdef (\Proof_\Sigma(\VerRandVec_{\NumVars-2}, \Query)_{\Query\in\Field}, \hat{F}(\VerRandVec,\Query)_{\Query\in \Field},\\ \ProofLD(\VerRandVec, \Query)_{\Query\in\Field}, \ProofLD(\Query, \revvec{\VerRand})_{\alpha \in \Field})~.
    \end{equation*}

    We show next that the verifier's view of the corrected proof $\LDView(\VerRandVec)$ is far (in expectation) from any accepting view. We accomplish this step by appealing to the robustness of the sumcheck PCPP (\cref{construc:robust-sumcheck}), and by arguing that $\Verifier_1$ inherits this robustness. Then as the view of the actual proof $\HonView(\VerRandVec)$ is close (in expectation) to the view of the corrected proof $\LDView(\VerRandVec)$, $\HonView(\VerRandVec)$ must be far (in expectation) from any accepting view.
    
    By \cref{lem:robust-sumcheck-tech}, if $\Verifier_\Sigma$ has oracle access to the evaluation table of a function $G\colon \Field^\NumVars \to \Field$, which is $\SumDist$-far from $\SCLang[\Field, \NumVars, \Degree, \Subcube, \SumcheckTotal]$, but $\RMDist$-close to $\ReedMuller[\Field, \NumVars, \Degree]$, then for all proofs $\Proof_\Sigma$, we have
    \begin{equation}
        \label{eqn:exp-rob-sumcheck}\mathop{\mathbb{E}}_{\VerRandVec\gets \Field^{\NumVars-1}}\left[\HammingDist(\SCView(\VerRandVec), \VAcc(\Verifier_\Sigma))\right] \geq \frac{1}{2}\min\left(\RMDist, 1-4\RMDist\right) = \frac{1}{2}\RMDist~,
    \end{equation} 
    where the final equality holds since $\RMDist \leq 1/5$. Now consider the verifier $\Verifier_1$, which uses its access to $F$ and $\Proof_P$ to run $\Verifier_{\Sigma}$ on input oracle $F+R$ and proof $\Proof_\Sigma$. Each location of $(F, \Proof_P)$ is used to compute at most two locations of $F+R$. This means that if $\Verifier_\Sigma$'s view of $F + R$ is an absolute distance of $k$ away from accepting, then $\Verifier_1$'s view of $(F,\Proof_P)$ must be at least an absolute distance of $k/2$ away from accepting in the worst case.
    
    Combining this observation with the fact that the full view $\LDView(\VerRandVec)$ of $\Verifier_1$ is twice the length of the full view $\SCView(\VerRandVec)$ of $\Verifier_\Sigma$, we see that the expected relative distance degrades by at most a factor of $1/4$. Thus, for all proofs $\Proof_\Sigma$,
    \begin{equation*}
        \mathop{\mathbb{E}}_{\VerRandVec\gets \Field^{\NumVars-1}}\left[\HammingDist(\LDView(\VerRandVec), \VAcc(\Verifier_1))\right] \geq \frac{1}{8}\RMDist.
    \end{equation*}
   Next, we upper bound the expected distance of $\HonView(\VerRandVec)$ and $\LDView(\VerRandVec)$. Since $\HammingDist(\ProofLD, \Proof_P) \leq \LDTProx$, we have that 
    \begin{equation*}
        \mathop{\mathbb{E}}_{\VerRandVec\gets \Field^{\NumVars-1}}\left[\HammingDist(\ProofLD(\VerRandVec, \Query)_{\Query \in \Field}, \Proof_P(\VerRandVec, \Query)_{\Query\in \Field})\right] \leq \LDTProx = \frac{1}{8}\RMDist.
    \end{equation*}
    Since $\ProofLD(\VerRandVec, \alpha)_{\alpha \in \Field}$ and $\Proof_P(\VerRand,\alpha)_{\alpha \in \Field}$ are substrings of $\LDView(\VerRandVec)$ and $\HonView(\VerRandVec)$ (respectively) of half the length, and the other components of the two views are identical, it follows that
    \begin{equation*}
        \mathop{\mathbb{E}}_{\VerRandVec\gets \Field^{\NumVars-1}}\left[\HammingDist\left(\LDView(\VerRandVec), \HonView(\VerRandVec) \right)\right] \leq \frac{1}{16}\RMDist.
    \end{equation*}
    
    Finally, we put everything together to obtain the result:
    \begin{align*}
        \mathop{\mathbb{E}}_{\VerRandVec\gets \Field^{\NumVars-1}}\left[\HammingDist\left(\HonView(\VerRandVec), \VAcc(\Verifier_1)
        \right)\right] &\geq
        \mathop{\mathbb{E}}_{\VerRandVec\gets \Field^{\NumVars-1}}\left[\left|\HammingDist\left(\LDView(\VerRandVec), \VAcc(\Verifier_1) \right) - \HammingDist\left(\HonView(\VerRandVec),\LDView(\VerRandVec)
        \right)\right|\right]
        \\
        & \geq \left|\mathop{\mathbb{E}}_{\VerRandVec\gets \Field^{\NumVars-1}}\left[ \HammingDist\left (\LDView(\VerRandVec), \VAcc(\Verifier_1) \right) \right] -  \mathop{\mathbb{E}}_{\VerRandVec\gets \Field^{\NumVars-1}}\left[\HammingDist\left(\HonView(\VerRandVec),\LDView(\VerRandVec)
        \right) \right]
        \right|
        \\
        &\geq \frac{1}{8}\RMDist - \frac{1}{16}\RMDist
        \\
        &= \frac{1}{16}\RMDist,
    \end{align*}
    where the first inequality is a consequence of the reverse triangle inequality, and the second is due to linearity of expectation and the fact that the $\mathbb{E}|X| \geq |\mathbb{E}[X]|$ for any random variable $X$.
\end{proof}

We now prove \cref{lem:rob-zk-pcp-sumcheck}. We prove robustness by splitting the argument into two cases: one in which the input is far from even being low-degree (in which case the low-degree test guarantees robustness), and the other in which the input is close to low-degree, but far from summing to the correct value (in which case \cref{lem:pzk-sumcheck-tech} guarantees robustness). As alluded to earlier, we prove zero knowledge using the fact that this robust PCP is locally computable from the PZK-PCP constructed in \cite{GurOS2024}.

\begin{proof}[Proof of \cref{lem:rob-zk-pcp-sumcheck}]
    Completeness is straightforward. Here we argue robust soundness and perfect zero-knowledge.

    \parhead{Robust soundness} Suppose that $\BadF$ is $\Proximity$-far from  $\SCLang[\Field, \NumVars, \Degree, \Subcube, \SumcheckTotal]$. First, if $\BadF$ is $\RMDist$-far from $\ReedMuller[\Field, \NumVars, \Degree]$, then by the robustness of the low-degree test (\cref{thm:bundled-ldt}), the expected distance of the verifier's view of $\BadF$ to any accepting view is $\Omega(\RMDist) = \Omega(\Proximity)$. As the low-degree test makes $O(|\Field|)$ many queries to $\BadF$, and the verifier makes $O(|\Field|)$ many queries to both $\Proof_\Sigma$ and $\Proof_P$, this forms a constant fraction of the view of the verifier. Therefore the expected distance of the view of the verifier from any accepting view is also $\Omega(\Proximity)$.

    Second, if $\BadF$ is $\Proximity$-far from $\SCLang[\Field, \NumVars, \Degree, \Subcube, \SumcheckTotal]$, but $\RMDist$-close to $\ReedMuller[\Field, \NumVars, \Degree]$, then by \cref{lem:pzk-sumcheck-tech}, the view of the verifier (when ignoring the low-degree test) to any accepting view is $\Omega(\RMDist) = \Omega(\Proximity)$. Once again, this forms a constant fraction of the number of queries made by the verifier overall, so the expected distance of the view of the verifier from any accepting view is $\Omega(\Proximity)$. In other words, \cref{construc:robust-sumcheck-pzkpcp} has expected robustness $\Omega(\Proximity)$.

    Thus by \cref{prop:exp-rob-to-rob}, for an appropriate choice of constant, \cref{construc:robust-sumcheck-pzkpcp} has robustness $\Omega(\Proximity)$.

    \parhead{Zero knowledge} We show that \cref{construc:robust-sumcheck-pzkpcp} is perfect zero knowledge by showing that it is $(\NumVars+1)$-locally computable from the perfect zero-knowledge PCP described in Construction 8.9 of \cite{GurOS2024}. The proof generated by \cref{construc:robust-sumcheck-pzkpcp} is $(\NumVars+1)$-locally computable by the following function:
    \begin{equation*}
        f^{(\Proof'_{\Sigma}, \Proof'_{Q}, \Proof'_{T_{1}},\dots,\Proof'_{T_{\NumVars}})}(\Oracle, \vec{\Query}) = \begin{cases}
            (\Proof'_{Q}(\vec{\Query}), \Proof'_{T_{1}}(\vec{\Query}),\dots,\Proof'_{T_{\NumVars}}(\vec{\Query}))~&\text{if}~\Oracle = \Proof_P
            \\
            (\Proof'_{\Sigma}(\Query_{\NumVars-1}), \Proof'_{\Sigma}(\Query_1,\Query_{\NumVars-1}), \dots, \Proof'_{\Sigma}(\vec{\Query}_{i-1}, \Query_{\NumVars-1}), \dots, \Proof'_{\Sigma}(\vec{\Query}))~&\text{if}~\Oracle = \Proof_{\Sigma}
        \end{cases}
        ,
    \end{equation*}
    where $(\Proof'_{\Sigma}, \Proof'_{Q}, \Proof'_{T_{1}},\dots,\Proof'_{T_{\NumVars}})$ is a proof sampled by the prover of Construction 8.9 of \cite{GurOS2024}. Thus, by \cref{lem:locally-computable-implies-zk} (in this work), and Lemma 8.11 from \cite{GurOS2024}, \cref{construc:robust-sumcheck-pzkpcp} is perfect zero knowledge.
    
    \parhead{Efficiency} The proof length and verifier complexity are clear from construction. For every query that $\Verifier_\Sigma$ makes to its input, the verifier in \cref{construc:robust-sumcheck-pzkpcp} makes three queries to $F$ and $\Proof_P$. For every query $\Verifier_\Sigma$ makes to its proof, the verifier in \cref{construc:robust-sumcheck-pzkpcp} makes one query to $\Proof_\Sigma$. By \cref{lem:rob-zk-pcp-sumcheck}, $\Verifier_\Sigma$ makes $O(|\Field|)$ queries to its input and $|\Field|$ queries to $\Proof$, and the desired query complexity follows.
\end{proof}

\section{Robust PZK-PCPs for NP and NEXP}
\label{sec:nexp}

In this section we show how to build robust PZK-PCPs for NP and NEXP using a robust PZK-PCPP for sumcheck as a building block.

\begin{definition}[Oracle 3-SAT]
	Let $B \colon \{0,1\}^{r + 3s + 3} \to \{0,1\}$ be a 3-CNF. We say that $B$ is \emph{implicitly satisfiable} if there exists $A \in \Bits^s \to \Bits$ such that for all $z \in \Bits^r, b_1,b_2,b_3 \in \Bits^s$, $B(z,b_1,b_2,b_3,A(b_1),A(b_2),A(b_3)) = 1$. Let $\OSAT$ be the language of implicitly satisfiable 3-CNFs.
\end{definition}

\begin{theorem}[Cook-Levin]
	For any $T(n)$-time nondeterministic Turing machine $M$, $T = \Omega(n)$, there is a polynomial-time reduction $R_M$ such that for any input $x \in \{0,1\}^n$,
	\begin{equation*}
		R_M(x,T) \in \OSAT \Leftrightarrow \exists w, M(x,w) = 1~.
	\end{equation*}
	Moreover, $R_M(x,T)$ is a formula in $O(\log T(n))$ variables of size $\Poly(n,\log T(n))$.
\end{theorem}

\begin{theorem}
\label{thm:robust-pzk-pcp-np-nexp}
	The following inclusions hold:
	\begin{enumerate}[label=(\roman*)]
		\item \label{item:pzk-pcp-np} for any $\QueryBound \leq 2^{\Poly(n)}$, $\NP \subseteq \RobustPZKPCP_{\QueryBound, \Alphabet(n)}[\log n + \log \QueryBound, \Poly(\log n + \log \QueryBound)]$, where $|\Alphabet(n)| = \Poly(n, \QueryBound)$, and
		\item \label{item:pzk-pcp-nexp} for any $\QueryBound \leq 2^{\Poly(n)}$, $\NEXP \subseteq \RobustPZKPCP_{\QueryBound, \Alphabet(n)}[\Poly(n)+ \log \QueryBound, \Poly(n)]$, where $|\Alphabet(n)| = \Poly(2^n, \QueryBound)$.
	\end{enumerate}
\end{theorem}
\begin{proof}
	Fix a language $\Language \in \NTIME(T)$, where for \ref{item:pzk-pcp-np}, $T = \Poly(n)$, and for \ref{item:pzk-pcp-nexp}, $T = 2^{\Poly(n)}$. Let $\Field$ be a field of characteristic $2$ of size to be fixed, $H$ a subfield of $\Field$ of size $\Theta(\log T)$. We will fix an arithmetisation $\hat B$ of $1-B$ over $\Field$; i.e., for all $c \in \Bits^{r + 3s + 3}$, $\hat{B}(c) = 1-B(c)$. The PCP constructions which prove \ref{item:pzk-pcp-np} and \ref{item:pzk-pcp-nexp} differ only in the choice of $\hat{B}$. For \ref{item:pzk-pcp-np}, note that since $r + 3s + 3 = O(\log n)$, the evaluation table of $B$ is computable in polynomial time. Hence we choose $\hat{B}$ to be the unique multilinear extension of $\hat{B}$; this can be evaluated in polynomial time, and the total degree $d_B$ of $\hat{B}$ is $O(\log n)$. For \ref{item:pzk-pcp-nexp}, we will take $\hat B$ to be the function computed by the arithmetic formula corresponding to the boolean formula $B$; this has total degree $d_B = O(|B|) = \Poly(n)$. In both cases we choose $\Field$ so that $d_B \ll |\Field| = \Poly(\log T)$.
	
	Let $m_1 \eqdef r/\log |H|$, $m_2 \eqdef s/\log |H|$. Let $\gamma_1 \colon H^{m_1} \to \Bits^r$, $\gamma_2 \eqdef H^{m_2} \to \Bits^s$ be the lexicographic orderings of the elements of $H^{m_1},H^{m_2}$ respectively. By \cite[Claim 4.2]{GoldwasserKR15}, for $i \in [2]$ the unique minimal-degree extension $\hat\gamma_i$ of $\gamma_i$ has individual degree $|H|-1$ and can be evaluated in time $\Poly(|H|,m_i,\log |\Field|)$.
	
	For a polynomial $\hat A \colon \Field^{m_2} \to \Field$, we define $g_{\hat{A}} : \Field^{m_1+3m_2+3} \to \Field$:
	\begin{equation*}
		g_{\hat{A}}(z,b_1,b_2,b_3,a_1,a_2,a_3) \eqdef \hat{B}(\gamma_1(z),\gamma_2(b_1),\gamma_2(b_2),\gamma_2(b_3),a_1,a_2,a_3) \cdot \prod_{i=1}^{3} (\hat{A}(b_i)+a_i-1)~.
	\end{equation*}
	Next, for a polynomial $\hat C \colon \Field^{m_2 + \ZKVars} \to \Field$, we define\footnote{Here, since $H$ is a subfield of $\Field$, we cannot divide $a_i-1$ by $|H|^k$ as suggested in \cref{sec:tech}. Instead we use a Lagrange polynomial so that the term $a_i-1$ appears only once in the sum over $H^k$.} $h_{\hat{C}} : \Field^{m_1+3m_2+3+3\ZKVars} \to \Field$:
	\begin{align*}
		h_{\hat{C}}(z,b_1,b_2,b_3,a_1,a_2,a_3,c_1,c_2,c_3) \eqdef \hat{B}(\gamma_1(z),\gamma_2(b_1),\gamma_2(b_2),&\gamma_2(b_3),a_1,a_2,a_3) \\ &\cdot \prod_{i=1}^{3} \left(\hat{C}(b_i,c_i)+\LagrangePoly{H^k,0^k}(c_i) \cdot (a_i-1)\right)~.
	\end{align*}
	Observe that $h_{\hat{C}}$ is a polynomial of total degree $\Degree = O((d_B+k) \cdot |H| + \deg(C))$ in $O(k + \log T/\log \log T)$ variables, and that if $\hat{A}(X) \eqdef \sum_{c \in H^\ZKVars} \hat{C}(X,c)$,
	
	\begin{equation}
		\label{eq:sum-g-h}
		\sum_{c_1,c_2,c_3 \in H^\ZKVars} h_{\hat{C}}(z,b_1,b_2,b_3,a_1,a_2,a_3,c_1,c_2,c_3) \equiv g_{\hat A}(z,b_1,b_2,b_3,a_1,a_2,a_3).
	\end{equation}
	
	We are now ready to specify the construction. First, we choose the parameter $k$. In the construction, $\deg(C)$ will be $O((\log T + k)|H|)$. Denote by $\Simulator'$ the simulator guaranteed by \cref{lem:rob-zk-pcp-sumcheck}. Let $\tilde{q}$ be a polynomial such that, for all $((\Field, \NumVars, \Degree, \Subcube, \gamma), F) \in \SCLang$, and for all PCPP verifiers $\MalVerifier$ that make $\QueryBound$ queries to the input $F$ and the proof, $\Simulator'$ makes $p(\QueryBound, \NumVars, \Degree, |H|, \log |\Field|)$ queries to $F$. Define $\tilde q \eqdef \tilde{q}(\QueryBound, \NumVars_1 + 3 \NumVars_2 + 3 + \ZKVars, \Degree, |H|, \log |\Field|) = \Poly(\QueryBound, \log T, k)$. Choose $\ZKVars = O(\log q^* + \log \log T)$ large enough such that $\ZKVars > \log \tilde q / \log |H|$.
	
	\begin{mdframed}[nobreak=true]
		\begin{construction}
            \label{cons:pzkpcp}
			A PZK-PCP for $\OSAT$.
			
			\noindent\textbf{Proof:}
			\begin{enumerate}
				\item \label{step:pzkpcp-prover-c} Let $A \colon \Bits^s \to \Bits$ be a satisfying assignment for $B$. Choose a uniformly random polynomial $\hat{C} \colon \Field^{m_2 + \ZKVars}$ of individual degree $2(|H|-1)$ such that for all $b \in H^{m_2}$, $\sum_{c \in H^\ZKVars} \hat C(b,c) = A(\gamma_2(b))$. Compute the full evaluation table $\Proof_C$ of $\hat C$.
				\item For each $\vec \tau \in \Field^{m_1+3m_2+3}$, let $\Proof_{\vec \tau}$ be a PZK-PCPP (\cref{lem:rob-zk-pcp-sumcheck}) for the claim
				\begin{equation}
                        \label{eqn:sc-claim}
					\sum_{\substack{z \in H^{m_1} \\ b_1,b_2,b_3 \in H^{m_2}}} \sum_{a_1,a_2,a_3 \in \Bits} \sum_{c_1,c_2,c_3 \in H^{\ZKVars}} \LagrangePoly{H^{m_1+3m_2} \times \Bits^3,(z,\vec{b},\vec{a})}(\vec\tau) \cdot h_{\hat C}(\vec\tau,c_1,c_2,c_3) = 0~.
				\end{equation}
                    \item Output $(\Proof_C, (\Proof_{\vec\tau})_{\vec\tau \in \Field^{\NumVars_1+3\NumVars_2+3}})$.
			\end{enumerate}
			
			\noindent\textbf{Verifier:}
			\begin{enumerate}
				\item \label{step:pzk-pcp-v-ldt} Perform a low total degree test (\cref{thm:bundled-ldt}) on $\Pi_C$ with $\varepsilon = 1/100$ and $\delta = 1/2$. If the test rejects, then reject.
				\item Choose $\vec\tau = (\zeta,\nu_1,\nu_2,\nu_3,\xi) \in \Field^{m_1} \times (\Field^{m_2})^3 \times \Field^{3}$ uniformly at random.
				\item \label{step:pzk-pcp-v-sc} Run the verifier for the sumcheck PZK-PCPP on $\Proof_{\vec\tau}$, yielding a claim 
					\begin{equation}
						\label{eq:zkpcp-sc-claim}
					\LagrangePoly{H^{m_1+3m_2} \times \Bits^3,\beta}(\vec\tau) \cdot h_{\hat C}(\vec\tau,\eta_1,\eta_2,\eta_3) = \gamma
					\end{equation}
					for some $\beta \in \Field^{m_1+3m_2+3}, \eta_1,\eta_2,\eta_3 \in \Field^{3\ZKVars}$, $\gamma \in \Field$. Query $\hat{C}$ at $(\nu_i,\eta_i)$ for $i \in \{1,2,3\}$ in order to compute $h_{\vec C}(\vec\tau,\eta_1,\eta_2,\eta_3)$. Accept if \eqref{eq:zkpcp-sc-claim} is true, else reject.
			\end{enumerate}

                \noindent We ``balance'' the verifier's queries by repeating \cref{step:pzk-pcp-v-ldt,step:pzk-pcp-v-sc} sufficiently many times so that each test accounts for at least $1/3$ of the verifier's view.
		\end{construction}
	\end{mdframed}

        \parhead{Proof length, alphabet size and query complexity}
        The proof length is dominated by the sumcheck proofs $\Proof_{\vec \tau}$; each such proof is of size $|\Field|^{O(\NumVars_1+\NumVars_2+\ZKVars+3)} = \Poly(T(n),\QueryBound)$, and there are $|\Field|^{O(\NumVars_1+\NumVars_2+3)} = \Poly(T(n))$ of them. The proof alphabet is $\Field^{\NumVars_1 + \NumVars_2 + k + 4}$, so the alphabet is of size $|\Field|^{O(\NumVars_1+\NumVars_2+\ZKVars+3)} = \Poly(T(n),\QueryBound)$. The query complexity is $\Poly(\log T(n), \log \QueryBound)$.
	
	\parhead{Completeness and robust soundness}
	Completeness and soundness rely on the following key observation.
	
	\begin{claim}
		\label{claim:osat-lde}
		The following are equivalent:
		\begin{enumerate}[label=(\roman*),noitemsep]
			\item \label{item:b-in-o3sat} $B \in \OSAT$;
			\item \label{item:a-sat-g} there exists a polynomial $\hat A \colon \Field^\NumVars \to \Field$ such that for all $z \in H^{m_1},b_1,b_2,b_3 \in H^{m_2},a_1,a_2,a_3 \in \Bits$, $g_{\hat{A}}(z,b_1,b_2,b_3,a_1,a_2,a_3) = 0$;
			\item \label{item:c-sat-h} there exists a polynomial $\hat C \colon \Field^{\NumVars_2+\ZKVars} \to \Field$ such that for all $z \in H^{m_1},b_1,b_2,b_3 \in H^{m_2},a_1,a_2,a_3 \in \Bits$, $\sum_{c_1,c_2,c_3 \in H^{\ZKVars}} h_{\hat{C}}(z,b_1,b_2,b_3,a_1,a_2,a_3,c_1,c_2,c_3) = 0$.
		\end{enumerate}

	\end{claim}
	\begin{proof}
		The equivalence of \ref{item:c-sat-h} and \ref{item:a-sat-g} is a direct consequence of \cref{eq:sum-g-h}.  It remains to prove that \ref{item:b-in-o3sat} and \ref{item:a-sat-g} are equivalent.
		
		Note first that for any $z \in \Bits^r, b_1,b_2,b_3 \in \Bits^s, a_1,a_2,a_3 \in \Bits$, $g_{\hat{A}}(z,b_1,b_2,b_3,a_1,a_2,a_3) = 0$ if and only if either $B(z,b_1,b_2,b_3,a_1,a_2,a_3) = 1$  (since $\hat B$ extends $1-B$) or, for some $i \in [3]$, $\hat{A}(b_i) = 1 - a_i$.
		
		\parhead{\ref{item:b-in-o3sat} $\Rightarrow$ \ref{item:a-sat-g}}
		Suppose that $B \in \OSAT$, and let $\hat{A}$ be the multilinear extension of a satisfying assignment $A$, and fix $z,b_1,b_2,b_3,a_1,a_2,a_3$. If $B(z,b_1,b_2,b_3,a_1,a_2,a_3) = 1$ then we are done, so suppose not. Then since $A$ is a satisfying assignment, there exists $i \in [3]$ such that $\hat{A}(b_i) = A(b_i) \neq a_i$, so $\hat{A}(b_i) = 1-a_i$ since $A(b_i) \in \{0,1\}$.
		
		\parhead{\ref{item:a-sat-g} $\Rightarrow$ \ref{item:b-in-o3sat}}
		Fix some polynomial $\hat A$ such that $g_{\hat A}$ is zero on $\Bits^{r+3s+3}$. Let $A$ be the assignment given by, for each $b \in \Bits^s$,
		\[
		A(b) = \begin{cases}
			\hat{A}(b) & \text{if $\hat{A}(b) \in \Bits$, or} \\
			0 & \text{otherwise.}
		\end{cases}
		\]
		Fix $z \in \Bits^r, b_1,b_2,b_3 \in \Bits^s$; we show that $B(z,b_1,b_2,b_3,A(b_1),A(b_2),A(b_3)) = 1$. Indeed, let $a_1,a_2,a_3 \in \Bits$ be such that $B(z,b_1,b_2,b_3,a_1,a_2,a_3) = 0$; if no such assignment exists then we are already done. Since $B(z,b_1,b_2,b_3,a_1,a_2,a_3) = 0$, it must be that there exists $i \in [3]$ such that $\hat{A}(b_i) = 1 - a_i$; and since $1-a_i \in \{0,1\}$, $A(b_i) = \hat{A}(b_i)$. Hence $A(b_i) \neq a_i$. Thus if $A(b_i) = a_i$ for all $i$, then $B(z,b_1,b_2,b_3,a_1,a_2,a_3) = 1$, from which the implication follows.
	\end{proof}
	
	For completeness, it suffices to observe that the claim \eqref{eqn:sc-claim} is true for $\hat{C}$ as in \cref{claim:osat-lde}. We proceed to show robust soundness. Suppose that $B \notin \OSAT$.
	
	Let $\varepsilon \eqdef 1/100$. First, suppose that $\Proof_C$ is $\varepsilon$-far from $\ReedMuller[\Field,\NumVars_2+\ZKVars,(\NumVars_2+\ZKVars)\Degree]$. Then the verifier's view in the low-degree test is on average $\Omega(\varepsilon)$-far from accepting by \cref{thm:bundled-ldt}.
	
	Otherwise, let $\hat{C} \in \ReedMuller[\Field,\NumVars_2+\ZKVars]$ be the unique closest codeword to $\Proof_C$; then $\Distance(\hat{C},\Proof_C) \leq \varepsilon$. It follows that $\Distance(h_{\hat{C}}, h_{\Proof_C}) \leq 3\varepsilon$, since a uniformly random evaluation of $h_{\hat{C}}$ depends on three uniformly random evaluations of $\hat{C}$.

	Since $B \notin \OSAT$, by \cref{claim:osat-lde} there exists $z \in H^{\NumVars_1}, b_1,b_2,b_3 \in H^{\NumVars_2}, a_1,a_2,a_3 \in \{0,1\}$ such that $\sum_{c_1,c_2,c_3 \in H^{\ZKVars}} h_{\hat{C}}(z,b_1,b_2,b_3,a_1,a_2,a_3,c_1,c_2,c_3) \neq 0$. It follows that
				\begin{equation*}
					\sum_{\substack{z \in H^{m_1} \\ b_1,b_2,b_3 \in H^{m_2}}} \sum_{a_1,a_2,a_3 \in \Bits} \LagrangePoly{H^{m_1+3m_2} \times \Bits^3,(z,\vec{b},\vec{a})}(\vec X) \cdot \sum_{c_1,c_2,c_3 \in H^{\ZKVars}} h_{\hat C}(\vec X,c_1,c_2,c_3)
				\end{equation*}
		is a nonzero polynomial in $\NumVars_1+\NumVars_2+3$ variables $\vec{X}$ of individual degree $O(d_B \cdot |H|)$. Thus, with probability $1-O((\NumVars_1+\NumVars_2+3)d_B \cdot |H|/|\Field|)$ over the choice of $\vec{\tau}$, the claim in \cref{eq:zkpcp-sc-claim} is false. By the expected robustness guarantee of the PZK-PCPP (\cref{lem:rob-zk-pcp-sumcheck}), the verifier's view of the sumcheck PCP in this case is $\Omega(\varepsilon)$-far on average from an accepting view. By a union bound, the overall expected robustness is $\Omega(\varepsilon)$.
	
	\parhead{Zero knowledge}
        Our simulator makes use of an algorithm due to Ben-Sasson et al. \cite{BenSassonCFGRS17} for efficiently lazily sampling random multivariate polynomials.

\newcommand{\PolySimAlgorithm}{\mathsf{PolySim}}
        \begin{lemma}[\cite{BenSassonCFGRS17}, Corollary 4.10]
\label{lemma:polysim}
There exists a probabilistic algorithm $\PolySimAlgorithm$ such that, for every finite field $\Field$, $\NumVars,\Degree \in \N$, set $S = \{(\alpha_{1},\beta_{1}), \dots, (\alpha_{\ell}, \beta_{\ell})\} \subseteq \Field^{\NumVars} \times \Field$, and $(\alpha,\beta) \in \Field^{\NumVars} \times \Field$,
\begin{equation*}
\Pr\Big[
\PolySimAlgorithm(\Field,\NumVars,\Degree,S,\alpha) = \beta
\Big]
=
\Pr_{\RandPoly \gets \Field^{\leq \Degree}[X_1,\ldots,X_\NumVars]}
\left[
\RandPoly(\alpha) = \beta
\pST
\begin{array}{c}
\RandPoly(\alpha_{1}) = \beta_{1} \\
\vdots \\
\RandPoly(\alpha_{\ell}) = \beta_{\ell}
\end{array}
\right]\enspace.
\end{equation*}
Moreover $\PolySimAlgorithm$ runs in time $\Poly(\log |\Field|, \NumVars, \Degree, \ell)$.
\end{lemma}

        We describe the construction of the simulator below, and then prove that its output is distributed identically to $\View_{\MalVerifier,\Prover}$.
        \begin{mdframed}
            \begin{construction}
                A simulator $\Simulator$ for \cref{cons:pzkpcp}.

                \noindent On query $\vec \Query$ to oracle $\Oracle$:
                \begin{enumerate}
                    \item \label{item:sim-c} If $\Oracle = \Proof_C$, sample $\beta \gets \PolySimAlgorithm(\Field,\NumVars+\ZKVars,\Degree,S,\vec{q})$, where $S$ is the set of prior queries to $\Proof_C$ and their answers, and respond with $\beta$.
                    
                    \item If $\Oracle = \Proof_{\vec \tau}$ for some $\vec \tau \in \Field^{\NumVars_1+3\NumVars_2+3}$:
                    \begin{enumerate}[nolistsep]
                    	\item If this is the first query to $\Proof_{\vec \tau}$, start a new instance $\Simulator'_{\vec \tau}$ of the  sumcheck PZK-PCPP simulator (\cref{lem:rob-zk-pcp-sumcheck}).
                    	\item Answer the query using $\Simulator'_{\vec \tau}$. This sub-simulator may make queries to the summand polynomial \eqref{eqn:sc-claim}. Each such query can be efficiently answered by making $3$ queries to $\Proof_C$, which are answered as in \cref{item:sim-c}.
                    \end{enumerate}
                \end{enumerate}
            \end{construction}
        \end{mdframed}
 
 	\newcommand{\HybSim}{\widetilde{\Simulator}}
 
 	We prove perfect zero knowledge via a hybrid argument. Consider a hybrid oracle simulator $\HybSim$ which behaves as $\Simulator$ except that it answers queries to $\Proof_C$ by querying an external oracle.
 		
 	We consider the following sequence of hybrid algorithms:
 	\begin{itemize}[noitemsep]
 		\item $H_0$: $\Simulator^{\MalVerifier}(x)$.
 		\item $H_1$: $\HybSim^{\MalVerifier,Z}(x)$, where $Z \colon \Field^{\NumVars_2+\ZKVars}$ is a uniformly random polynomial of individual degree $|\Subcube|-1$.
 		\item $H_2$: $\HybSim^{\MalVerifier,Z}(x)$, where $Z \colon \Field^{\NumVars_2+\ZKVars}$ is sampled as $\hat{C}$ in \cref{step:pzkpcp-prover-c} of the prover algorithm.
 		\item $H_3$: $\View_{\MalVerifier,\Prover}(x)$.
 	\end{itemize}
 	
 	\begin{itemize}
 		\item \ul{$H_0 \equiv H_1$.} This follows directly from the correctness of $\PolySimAlgorithm$ (\cref{lemma:polysim}).
 		
 		\item \ul{$H_1 \equiv H_2$.}
 		This claim relies crucially on the following lemma, originally shown in \cite{ChiesaFS17,ChiesaFGS22}.
 
	\begin{lemma}[{\cite[Corollary 5.3]{ChiesaFS17}}]
		\label{lemma:sum-indep}
		Let $\Field$ be a finite field, $H \subseteq \Field$, $\Degree,\Degree' \in \N$ with $\Degree' \geq 2(|H|-1)$. Let $Q \subseteq \Field^{\NumVars+\ZKVars}$ with $|Q| < |H|^\ZKVars$. Let $Z$ be chosen uniformly at random in $\Field[X_1,\ldots,X_\NumVars,Y_1,\ldots,Y_\ZKVars]$ such that $\deg_{X_i}(Z) \leq \Degree$ for $1 \leq i \leq \NumVars$ and $\deg_{Y_i}(Z) \leq \Degree'$ for all $1 \leq i \leq \ZKVars$.
		
		The ensembles $(\sum_{\vec y \in H^\ZKVars} Z(\vec \alpha, \vec y))_{\vec \alpha \in \Field^\NumVars}$ and $(Z(\vec q))_{\vec q \in Q}$ are independent.
	\end{lemma}
	
	Since $\MalVerifier$ makes at most $\QueryBound$ queries to the proof, $\HybSim$ makes at most $p(\QueryBound)$ queries to its oracle. Then since $p(\QueryBound) < |H|^\ZKVars$ by assumption, by \cref{lemma:sum-indep}, the answers to those queries are statistically independent from $(\sum_{\vec y \in H^\ZKVars} Z(\vec \alpha, \vec y))_{\vec \alpha \in \Field^\NumVars}$. Thus for all $v$,
	\begin{align*}
		\Pr[H_2 \to v] &= \Pr_Z\big[\HybSim^{\MalVerifier,Z}(x) \to v ~\big\vert~ \forall b \in H^{m_2},\, \sum_{c \in H^\ZKVars} \hat C(b,c) = A(\gamma_2(b))\big] \\
		&= \Pr_Z\big[\HybSim^{\MalVerifier,Z}(x) \to v \big]~.
	\end{align*}
	
	\item \ul{$H_2 \equiv H_3$.} This follows from the perfect zero knowledge guarantee of the sumcheck PZK-PCPP. \qedhere
	\end{itemize}
\end{proof}

\section{PZK-PCP for NP and NEXP with constant query complexity}
In this section, we combine our zero-knowledge proof composition theorem and our randomness-efficient, robust PZK-PCPs for $\NP$ and $\NEXP$ to obtain constant-query PZK-PCPs for $\NP$ and $\NEXP$.

\begin{theorem}
    The following inclusions hold:
    \begin{enumerate}[label=(\roman*)]
        \item $\NP \subseteq \PZKPCP[\log n, 1]$, and
        \item $\NEXP \subseteq \PZKPCP[\Poly(n), 1]$.
    \end{enumerate}
    Moreover, when $\QueryBound$ is polynomial, $\NP$ has efficient PZK-PCPs with constant query complexity and query bound $\QueryBound$.
\end{theorem}

\begin{proof}
    \parhead{Part (i)} 
    We prove this via a chain of inclusions as follows. 
    
    We start by performing alphabet reduction on the robust PZK-PCP for $\NP$ we constructed in \cref{thm:robust-pzk-pcp-np-nexp}. Let $\QueryBound \leq 2^{\Poly(n)}$ be an arbitrary query bound; for simplicity, we will assume without loss of generality that $\QueryBound \geq n$. By part (i) of \cref{thm:robust-pzk-pcp-np-nexp}, for any $\tilde q \leq 2^{\Poly(n)}$, 
    \begin{equation*}
        \NP \subseteq \RobustPZKPCP_{\tilde{q}, \Alphabet(n)}[\log n + \log \tilde{q}, \Poly(\log n + \log \tilde{q})],
    \end{equation*}
    where $|\Alphabet(n)| = \Poly(n, \tilde{q})$. This can be viewed as a PCP over the alphabet $\Bits^a$, where $a \eqdef \log|\Alphabet(n)| = O(\log n + \log \tilde{q})$. Therefore, we can apply alphabet reduction (\cref{lem:alphabet-reduction}) to obtain
    \begin{equation*}
        \RobustPZKPCP_{\tilde{q}, \Alphabet(n)}[\log n + \log \tilde{q}, \Poly(\log n + \log \tilde{q})] \subseteq \RobustPZKPCP_{\tilde{q}, \Bits}[\log n + \log \tilde q, \Poly(\log n + \log \tilde{q})].
    \end{equation*}
    
    Next, we perform proof composition to obtain constant query complexity. Let $Q(n) = \Poly(\log n + \log \tilde q)$ be the query complexity of the boolean $\RobustPZKPCP$ and set $\tilde q(n)$ equal to a polynomial in $\QueryBound$ and $n$ large enough to ensure that $\tilde{q}(n)/Q(n) \geq \QueryBound(n)$ for all $n\in\N$; $\tilde q = O((q^*)^{1+\varepsilon})$ suffices for any $\varepsilon > 0$.
    Then by \cref{cor:zk-proof-comp}, $\RobustPZKPCP_{\tilde q, \Bits}[r,q] \subseteq \PZKPCP_{\tilde q / q, \Bits}[r+\log n, 1]$. Thus, we have that 
    \begin{equation*}
        \RobustPZKPCP_{\tilde{q}, \Bits}[\log n + \log \tilde q, \Poly(\log n + \log \tilde{q})] \subseteq \PZKPCP_{\QueryBound, \Bits}[\log n + \log \QueryBound, 1].
    \end{equation*}
    By chaining the above inclusions together we have that $\NP \subseteq \PZKPCP_{\QueryBound, \Bits}[\log n + \log \QueryBound, 1]$. Since $q^*$ was arbitrary, $\NP \subseteq \PZKPCP[\log n, 1]$.

    \parhead{Part (ii)} Let $\QueryBound \leq 2^{\Poly(n)}$ be an arbitrary query bound. By part (ii) of \cref{thm:robust-pzk-pcp-np-nexp}, for any $\tilde{q} \leq 2^{\Poly(n)}$, 
    \begin{equation*}
        \NEXP \subseteq \RobustPZKPCP_{\tilde{q}, \Alphabet(n)}[\Poly(n) + \log \QueryBound, \Poly(n)],
    \end{equation*}
    where $|\Alphabet(n)| = \Poly(2^n, \tilde{q})$. As before, this can be viewed as a PCP over the alphabet $\Bits^a$, where $a \eqdef \log|\Alphabet(n)| = O(n + \log\tilde{q})$. Therefore, we can apply alphabet reduction (\cref{lem:alphabet-reduction}) to obtain 
    \begin{equation*}
        \RobustPZKPCP_{\tilde{q}, \Alphabet(n)}[\Poly(n) + \log \QueryBound, \Poly(n)]\subseteq \RobustPZKPCP_{\tilde{q}, \Bits}[\Poly(n), \Poly(n, \log \tilde{q})].
    \end{equation*}

    As before, we perform proof composition to obtain constant query complexity. Let $Q(n) = \Poly(n)$ be the query complexity of the boolean $\RobustPZKPCP$ and set $\tilde{q} \eqdef \QueryBound(n)\cdot Q(n)$ for all $n \in \N$, so $\tilde{q} = \Poly(\QueryBound, n)$. Then by \cref{cor:zk-proof-comp}, $\RobustPZKPCP_{\hat{q}, \Bits}[r,q] \subseteq \PZKPCP_{\hat{q}/Q, \Bits}[r+\log n, 1]$, so 
    \begin{equation*}
        \RobustPZKPCP_{\tilde{q}, \Bits}[\Poly(n), \Poly(n, \log \tilde{q})] \subseteq \PZKPCP_{\QueryBound, \Bits}[\Poly(n) + \log \QueryBound, 1].
    \end{equation*}
    Since $\QueryBound$ was arbitrary, $\NP \subseteq \PZKPCP[\Poly, 1]$.
\end{proof}

\appendix

\clearpage

\printbibliography

\end{document}